\documentclass[conference]{IEEEtran}

\usepackage{cite}
\usepackage{amsmath,amssymb,amsfonts}
\usepackage{graphicx}
\usepackage{textcomp}
\def\BibTeX{{\rm B\kern-.05em{\sc i\kern-.025em b}\kern-.08em
    T\kern-.1667em\lower.7ex\hbox{E}\kern-.125emX}}

\usepackage{latexsym}
\usepackage{float}
\usepackage[mathscr]{eucal}
\usepackage{ifpdf}
\usepackage{cite}
\usepackage{graphicx}
\usepackage{algorithm}
\usepackage{algpseudocode}

\usepackage{xcolor}

\usepackage{braket}

\usepackage{hyperref}

\usepackage{subcaption}

\makeatletter
\newtheorem{theorem}{Theorem}
\newtheorem{lemma}{Lemma}
\newtheorem{corollary}{Corollary}
\newtheorem{proposition}{Proposition}

\newtheorem{definition}{Definition}

\newtheorem{remark}{Remark}

\newcommand{\etal}{\textit{et al. }}

\DeclareMathAlphabet{\mathsfit}{\encodingdefault}{\sfdefault}{m}{sl}
\SetMathAlphabet{\mathsfit}{bold}{\encodingdefault}{\sfdefault}{bx}{n}

\usepackage{soul}
\soulregister\cite7
\soulregister\citet7
\soulregister\citep7

\newcommand{\Tr}{\operatorname{Tr}}
\newcommand{\op}{\operatorname{op}}

\begin{document}
\bstctlcite{IEEEexample:BSTcontrol}

\title{
Riemannian Optimization for Multi-Player Quantum Games on Product Unitary Manifolds
}

\author{
	\IEEEauthorblockN{
		\uppercase{Alireza~Habibi}\IEEEauthorrefmark{1}\IEEEauthorrefmark{2}, 
		\uppercase{Setareh~Maghsudi}\IEEEauthorrefmark{2}\IEEEauthorrefmark{3}
	}
	\IEEEauthorblockA{
		\IEEEauthorrefmark{2}Faculty of Electrical Engineering and Information Technology, Ruhr University Bochum, Bochum, Germany, 
	}
	\IEEEauthorblockA{\IEEEauthorrefmark{3} Faculty of Computer Science, Ruhr University Bochum, Bochum, German
	} 
	\IEEEauthorblockA{\IEEEauthorrefmark{1}Corresponding author: Alireza~Habibi (email: alireza.habibi@ruhr-uni-bochum.de)}
	
}

\maketitle

\begin{abstract}
	Quantum game theory is an extension  of classical game theory that uses quantum principles in game theory. 
	The Eisert–Wilkens–Lewenstein (EWL) quantum game is an early example of the two-player classical Prisoner’s Dilemma transformed into a quantum Prisoner’s Dilemma.
	In the EWL game, the players choose pure quantum strategies represented by unitary matrices.  
	This extension can resolve the classical dilemma by enabling cooperative equilibrium with higher payoff.  
	In this paper, we first discuss the Extended EWL (EEWL) for multiplayer quantum games with mixed strategies. 
	In EEWL, each player controls a set of unitary operators as quantum actions and uses a classical mixed strategy over these actions. The payoffs are defined as expectation values of Hermitian reward operators acting on a shared quantum state, which is generated and measured according to the EEWL protocol.  
	We then propose the Unitary Strategy Matrix Exponential Algorithm (USMEA), a geometry-aware sequential algorithm for the EEWL mixed-strategy setting, in which each player jointly learns a trainable set of local unitary actions and the associated classical mixing probabilities.
	Thereby it acts as a learning-and-control layer for multi-agent quantum decision systems.
	We analyze the convergence properties of USMEA under standard smoothness and step-size conditions and validate the theory with numerical experiments.
	These results show how classical optimization methods can be systematically integrated into the design and analysis of engineered quantum strategic interactions.
\end{abstract}

\section{Introduction}
Decision-making is a fundamental challenge in different fields like economics, biology, management, and artificial intelligence. 
The typical decision-making process includes multiple steps, like analyzing available information, predicting possible outcomes, and selecting the best strategy from different choices. 
Game theory studies the decision-making process in which the payoff to a player depends on its own strategies as well as those of others. ~\cite{edwards1954theory,camerer2004advances}.
When a player chooses a single action, this is called the pure strategy. 
A mixed strategy means that a player selects between multiple strategies based on a probability distribution.  In competitive or uncertain situations where uncertainty exists, mixed strategies can offer a strategic advantage~\cite{von2007theory,fudenberg1991game}.

One of the most popular games is prisoner's dilemma. 
The classical prisoner's dilemma shows how two rational players might choose not to cooperate, even when mutual cooperation would result in a better outcome for both~\cite{tucker1950two,poundstone2011prisoner,rapoport2018prisoner}.
In the EWL quantum game, introducing entanglement and allowing players to use $2 \times 2$ unitary operators with two free parameters as pure quantum strategies change the nature of the equilibrium point in the prisoner's dilemma. The resulting equilibrium becomes more efficient and fairer for the players.
Under these conditions, the Nash equilibrium aligns with the Pareto optimal outcome, which resolves the dilemma that exists in the classical version of the game~\cite{eisert1999quantum}.
However, when players are allowed to use unitary operators with three free parameters, the situation changes. In this case, a Nash equilibrium no longer exists, since each action can be countered by an opposing move that reduces the expected payoff.
When the EWL game is extended to include mixed quantum strategies, the Nash equilibrium reemerges. This new equilibrium can yield a higher expected payoff than its classical counterpart~\cite{eisert2000quantum}.
In this scenario, in the mixed-strategy quantum setting, the equilibrium or fixed-point structure may fail to be isolated
Since then, the EWL game has been adapted for multiplayer games and integrated into larger quantum strategic environments. 
Researchers have also investigated alternative forms of quantum games beyond the initial formulation~\cite{benjamin2001multiplayer,flitney2002quantum,schmid2010experimental}. 
Flitney~\etal examined the impact of decoherence in the EWL quantum game and showed that the quantum advantage is diminished as decoherence increases~\cite{flitney2004quantum}.
Xu~\etal implemented the EWL quantum game in the IBM quantum computer quantum hardware and showed that the difference between classical simulation and quantum execution was approximately $1\%$ in this game~\cite{xu2022experimental}.
The EWL quantum game has been used in many different areas like finance, quantum communication, cryptography, algorithmic trading, cognitive modeling, and economics~\cite{holtfort2024quantum,li2014entanglement,ullah2021game,hanauske2010doves,khrennikov2020quantum,khan2021quantum,ikeda2022theory}.
With the growing interest in quantum computing and quantum information science in recent years, the exploration of quantum game theory has gained significant attention~\cite{perez2024game,bostanci2022quantum,preskill2018quantum}.
While quantum hardware continues to advance and become more accessible, theoretical research in this field is still important. This kind of study helps to design new protocols, cover fundamental quantum advantages, and guide the creation of future quantum algorithms and applications in quantum game theory.

In game theory, gradient-based methods are widely used to enable players to iteratively refine their actions.
However, when unitary matrices represent actions, optimization becomes challenging.
Classical gradient descent with constraints in Euclidean space is inefficient and difficult to scale for unitary matrices~\cite{manton2002optimization, luchnikov2021riemannian}.
In these cases, the optimization in the Riemannian manifold is a suitable tool to address the challenge.
Riemannian gradient descent projects gradients onto tangent spaces and maps updates back to the manifold.
Such a method has seen success in low-rank matrix learning, signal processing, and deep learning with orthogonality constraints~\cite{abrudan2005optimization}.
The application of this idea to quantum games with mixed strategies has not been investigated.
Most existing approaches use a fixed set of unitary actions or predefine the classical probabilities for mixing strategies~\cite{flitney2002introduction, eisert2000quantum, manton2002optimization}.

We consider an extension of EWL quantum game in which each player jointly learns a set of local unitary actions and a classical mixed strategy over these actions. 
This leads to a hybrid decision space given by probability simplices and a product of unitary manifolds. 
We then study USMEA as a geometry-aware sequential method for the coupled multi-agent dynamics induced by this formulation.
In addition, this formulation is a coupled multi-agent problem that is non-convex, generally non-monotone, and not, in general, a potential game. 
Moreover, due to coupling and symmetry, the equilibrium and fixed-point structure can be nontrivial, and the fixed points may not be isolated, which makes even local convergence analysis meaningful in this setting.

\textbf{Our contributions:}  
Our main contributions in this paper are as follows.
\begin{itemize}
    \item We formulate the EEWL framework for multiplayer quantum games with mixed strategies, where each player jointly learns a trainable set of local unitary actions and a classical mixing distribution over these actions.
    \item We propose USMEA as a geometry-aware sequential learning scheme for jointly updating the unitary actions and the mixing probabilities.
    \item We show how standard mixed-strategy equilibrium-existence arguments apply to the EEWL model, and we derive the corresponding player-wise gradients, smoothness bounds, and step-size conditions for the update dynamics.
    \item In the experimental results, we study the convergence behavior of USMEA.
\end{itemize}
In this work, all learning and optimization steps of USMEA are carried out on a classical computer.
From an engineering perspective,  the proposed framework can be viewed  as a learning-and-control layer sitting on top of a quantum system where multiple agents act locally on a shared entangled state~\cite{song2025fast,le2022dqra,abane2025entanglement}. 
The EEWL game shows  how local
unitary actions on a shared state can result to payoffs, while USMEA gives a systematic approach to adjust these actions and their mixing probabilities using Riemannian optimization.
This links quantum game theory to the development of multi-agent controllers and learning-based protocols in emerging quantum technologies.

In Section~\ref{sec:notations}, we introduce the notation and basic concepts
from quantum information, quantum game theory, and Riemannian manifold that we use in the rest of the paper.
Section~\ref{sec:model} presents the EEWL model for multi-player quantum games with a general strategy space.
Section~\ref{sec:usmea} introduces the USMEA update rules.
In Section~\ref{sec:theoretical}, we study the convergence and performance of the proposed algorithm. 
Section~\ref{sec:expriments} provides experimental results for the algorithm’s performance in quantum games.
In Section~\ref{sec:conclusion}, we summarize the main findings of this work, discuss their implications, and outline potential directions for future works.

In addition, for readers with a game-theoretic background and limited familiarity with quantum basics, we provided a concise overview of the quantum mechanics used in this paper in Supplementary Material, Section~\ref{sec:quantumbasics}.
The Supplementary Material also contains additional theorems and all proofs
in Section~\ref{append:proofs}, as well as detailed descriptions of the
experimental setups.

\section{Notations}
\label{sec:notations}
This section introduces the basic notations of quantum game and the Riemannian manifold structure for unitary operators.

\subsection{Quantum game notations}
\label{sec:Gquantumbasics}
Let $\mathcal{H} \cong \mathbb{C}^d$ be a finite-dimensional complex Hilbert space, where $d$ is the dimension of the Hilbert space~\cite{nielsen2010quantum}. 
For an $N$-player quantum game, each player $i$ controls a local Hilbert space $\mathcal{H}^{(i)}$ of dimension $d_i$.  
The global Hilbert space of this composite system is given by the tensor product of the local Hilbert spaces,
\begin{align}
\mathcal{H} = \bigotimes_{i=1}^N \mathcal{H}^{(i)}, \qquad d=\dim\mathcal{H} = \prod_{i=1}^N d_i,
\end{align}
where  $d_i=\text{dim}(\mathcal{H}^{(i)})$ is the dimension of local Hilbert space $i$.
Quantum operators are denoted by bold capital letters. In particular, $\mathbf{U}^{(i)} \in \mathcal{U}(d_i)$ denotes the unitary operator used by player $i$ as its pure quantum strategy, and its dimension is  $d_i\times d_i$. 
To reflect the tensor-product structure, we can reindex the matrix entries $O_{m,n}$ of operator $\mathbf{O}$ using multi-indices $m\mapsto (j_1,\dots,j_N)$ and $n\mapsto (j'_1,\dots,j'_N)$ where 
$j_k,j'_k=1,\dots,d_k$. In this case the operator becomes a $2N$-index tensor with entries $O_{j_1,\dots,j_N,j'_1,\dots,j'_N}$.

In this paper, we employ Positive Operator-Valued Measurements (POVMs).
Let us consider a set of possible measurement outcomes $\Omega$. For each outcome $\omega \in \Omega$, we associate a positive semi-definite operator $\mathbf{P}_\omega : \mathcal{H} \to \mathcal{H}$ satisfying the completeness relation $ \sum_{\omega \in \Omega} \mathbf{P}_\omega = \mathbf{I}$.
The probability of obtaining a specific measurement outcome $\omega$ when the system is in the quantum state described by the density matrix $\rho$ is given by $p_\omega = \Tr\left( \mathbf{P}_\omega \rho \right)$. In addition, the completeness relation ensures that the total probability across all outcomes sums to one ( $\sum_{\omega \in \Omega} p_\omega = 1$). 
In a quantum game, a set of outcomes is used to compute and allocate payoffs to players. 

Small capital letters are used for vectors. 
We use the shorthand $(\mathbf{j}) \equiv (j_1, j_2, \cdots, j_N)$, where $j_i$ denotes the index corresponding to the player $i$. The notation $(j'_i; \mathbf{j}_{-i}) = (j_1, \cdots, j'_i, \cdots, j_N)$ shows that the index $j_i$ for player $i$ has been replaced by $j'_i$.
The term $(\mathbf{j}_{-i})$ refers to the tuple $(j_1, j_2, \cdots, j_{i-1}, j_{i+1}, \cdots, j_N)$, which omits the entry for player $i$.
In addition, we use $d_{-i}$ as the product of all subsystem dimensions except for $i$, that is $ d_{-i} = \prod_{k \neq i} d_k$.
The softmax function of a vector $\mathbf{a}=(a_1,\dots,a_m)$ at temperature $T$ is defined as
\begin{align}
    \sigma_T(\mathbf{a})_j = \frac{e^{a_j/T}}{\sum_{k=1}^m e^{a_k/T}},
\end{align}
with $
0 \le \sigma_T(a)_j \le 1$ for all $j$.
The $\sigma_T(\mathbf{a})$ is a probability vector with 
$\sum_{j=1}^m \sigma_T(a)_j = 1$.

We use the superscript $(\text{s})$ to indicate a strategy profile of local strategies applied independently by each player in product space. 
\begin{subequations}
\begin{align}
\label{eq:strategy_profile}
    \mathbf{U}^{(s)} &= \mathbf{U}^{(1)} \otimes \mathbf{U}^{(2)} \otimes \cdots \otimes \mathbf{U}^{(N)},
\\
    \mathbf{U}^{(\text{s},i)}  &= \mathbf{I}^{(1)} \otimes  \cdots  \otimes  \mathbf{U}^{(i)}  \otimes \cdots \otimes \mathbf{I}^{(N)} \nonumber \\
    &\equiv \mathbf{U}^{(i)} \otimes \mathbf{I}^{(-i)},
\\
    \mathbf{U}_{}^{(\text{s},-i)}  &= \mathbf{U}^{(1)} \otimes  \cdots \otimes \mathbf{I}^{(i)} \otimes \cdots \otimes \mathbf{U}^{(N)},
\\
 (\mathbf{U}^{'(i)} ;\mathbf{U}_{}^{(\text{s},-i)} ) &= (\mathbf{U}^{(1)} \otimes  \cdots \otimes \mathbf{U}^{'(i)} \otimes \cdots \otimes \mathbf{U}^{(N)})\nonumber \\
    &\equiv \mathbf{U}^{'(i)}  \otimes \mathbf{U}_{}^{(\text{s},-i)} ,
\end{align}
\end{subequations}
where $\text{dim}(\mathbf{I}^{(i)})=\text{dim} (\mathbf{U}^{(i)})$.
For the complex conjugate of a number $a$, we use the notation $\bar{a}$.
The trace operation is defined as 
\begin{subequations}
\begin{align}
    \Tr(\mathbf{A}) &= \sum_{\mathbf{j}}   A_{\mathbf{j},\mathbf{j}}, \\
    \Tr_{-i}(\mathbf{A})\Big|_{j_i,j'_i} &=
     \sum_{\mathbf{j}_{-i}}
    A_{(j_i;\mathbf{j}_{-i}) , (j'_i;\mathbf{j}_{-i})} 
    ,
\end{align}
\end{subequations}
where  $\Tr_{-i}$ denotes the partial trace over all indices except for subsystem $i$.

For any Hermitian operator $\mathbf{A}\in\mathbb{C}^{d\times d}$ with real eigenvalues $\{\lambda_i\}_{i=1}^d$, 
the operator (spectral) norm and the Frobenius (Hilbert–Schmidt) norm are defined as follows (see, e.g.,~\cite{watrous2018theory,horn2012matrix,bhatia2013matrix}):
\begin{align}
\|\mathbf{A}\|_{\mathrm{op}} &= \sup_{\|\mathbf{v}\|_2=1}\|\mathbf{A}\mathbf{v}\|_2= \max_j |\lambda_j|, \\
\|\mathbf{A}\|_F &= \sqrt{\operatorname{Tr}(\mathbf{A}^\dagger \mathbf{A})}=\Big(\sum_{i=1}^d \lambda_j^2\Big)^{1/2},
\end{align}
where $\mathbf{v} \in \mathbb{C}^d$ is a unit vector.
These norms satisfy the standard inequalities
\begin{align}
\|\mathbf{A}\|_{\mathrm{op}} \le \|\mathbf{A}\|_F \le \sqrt{d} \|\mathbf{A}\|_{\mathrm{op}},\quad
\|\mathbf{AB}\|_F \le \|\mathbf{A}\|_{\mathrm{op}} \|\mathbf{B}\|_F.
\end{align}

\subsection{Riemannian Manifold}
A Riemannian manifold $(\mathcal M,g)$ is a finite-dimensional smooth manifold $\mathcal M$ of dimension $n$ (with tangent bundle ($T\mathcal M$) equipped with a Riemannian metric $g$. A Riemannian metric $g$ on $\mathcal M$ is a smoothly varying field of inner products
$x \mapsto g_x(\cdot,\cdot)$ on the tangent spaces $T_x\mathcal M$ for each $x \in \mathcal M$. 
In this paper, we consider $\mathcal M_{\mathcal U}$ to be the unitary group $\mathcal U(n)$ and equip it with the standard bi-invariant Riemannian metric derived from the Frobenius inner product. 
At a point $\mathbf{U} \in \mathcal{U}(d_i)$ , the tangent space is
$
T_U \mathcal{U}(d_i) = \{\mathbf{X}\mathbf{U} | \mathbf{X}^\dagger = -\mathbf{X}\},
$
where $\mathbf{X}$ is skew-Hermitian. 
For any skew-Hermitian $\mathbf{X}$, the curve $\gamma(t) = \exp(t\mathbf{X})\mathbf{U}$ (with $t\in\mathbb{R}$) traces a path along the manifold. In practice, all updates to unitary actions are performed using exponential retractions of the form $\exp(\eta \mathbf{X})\mathbf{U}$.
The authors recommend references \cite{absil2008optimization,wiersema2023optimizing,luchnikov2021riemannian} for a more comprehensive introduction.

\section{Model}
\label{sec:model}
In this section, we first provide the definition of the EEWL quantum game. 
Next, we define the best response (BR) and the Nash equilibrium.
Finally, we discuss the existence of Nash equilibrium points within this quantum game.

\begin{definition}
\label{def:eewl-game}
An EEWL quantum game is a tuple $\mathcal{Q} = \langle \mathcal{N}, \mathcal{H}, \mathcal{S} , r\rangle$ defined as,
\begin{itemize}
    \item $\mathcal{N} = \{1, 2, \ldots, N\}$: A finite set of $N$ players.
    \item $\mathcal{H} = \mathcal{H}^{(1)} \otimes \mathcal{H}^{(2)}\otimes\cdots \otimes \mathcal{H}^{(N)}$: The Hilbert space of the game. Each player $i \in \mathcal{N}$ access a local Hilbert space $\mathcal{H}^{(i)} \cong \mathbb{C}^{d_i}$.
    \item 
    $\mathcal{S}^{(i)}=\left\{ \mathbf{U}^{(i)}_{j_i}\right\}_{j_i=1}^{m_i}$: 
    The set of actions (pure strategies) available to player~$i$.
    Each action $\mathbf{U}^{(i)}_{j_i} \in \mathcal{U}(d_i)$ acts on $\mathcal{H}^{(i)}$.
    The $m_i \ge 1$ is the total number of actions that player~$i$ can choose.
    \item $\mathcal{S} = \mathcal{S}^{(1)} \otimes\mathcal{S}^{(2)} \otimes \cdots \otimes \mathcal{S}^{(N)}$: The joint strategy space. A pure joint quantum strategy profile is the tensor product, 
    \begin{align}
        \mathbf{U}^{(s)}_{\mathbf{j}} = \mathbf{U}^{(1)}_{j_1} \otimes \mathbf{U}^{(2)}_{j_2} \otimes \cdots \otimes \mathbf{U}^{(N)}_{j_N}.
    \end{align}
    \item The reward or payoff function $r$: Each player $i \in \mathcal{N}$ receives a payoff determined by an function $r^{(i)} : \rho_\mathbf{j} \to \mathbb{R}$. 
    The density matrix $\rho_\mathbf{j} = \mathbf{U}^{(s)}_{\mathbf{j}} \rho_0 \mathbf{U}^{(s)\dag}_{\mathbf{j}}$ is the final state of the system when all players use the joint strategy profile $\mathbf{U}^{(\text{s})}_{\mathbf{j}}$, and $\rho_0$ denotes the initial state of the quantum system.
    To calculate the payoff, one can use Hermitian payoff operator $\mathbf{R}_i : \mathcal{H} \to \mathcal{H}$, such that
    \begin{align}
        \label{eq:rewart-puremixedpayoff}
        r^{(i)}_{\mathbf{j}} &= \Tr \left(\mathbf{R}_i \rho_{\mathbf{j}}\right).
    \end{align}
\end{itemize}
\end{definition}

\textbf{Mixed Strategy Setup:}
When players adopt mixed strategies, they select among actions according to a classical probability distribution. This probabilistic choice is described by the vector
\begin{align}
\mathbf{p}^{(i)} = \left( p^{(i)}_1, p^{(i)}_2, \ldots, p^{(i)}_{m_i} \right)\in \Delta^{m_i-1},
\end{align}
where $\Delta^{m_i-1}$ is the $(m_i-1)$-dimensional probability simplex.
The set of (mixed quantum) strategies of player $i$
is
\begin{align}
   z^{(i)} =  \left( \{\mathbf{U}^{(i)}_{j_i}\},\mathbf{p}^{(i)}\right) = \left\{( \mathbf{U}_{j_i}^{(i)},p_{j_1}^{(i)}) | \sum_{j_i=1}^{m_i}p_{j_1}^{(i)}=1\right\}.
\end{align}
The joint mixed (quantum) strategies for all players is given by 
$z=(z^{(1)},\cdots,z^{(N)})$, 
which abbreviates the collection of all action sets and probability distributions.
We use 
 $z^{(-i)}=(\mathbf{U}^{(-i)},\mathbf{p}^{(-i)})$  as the collection of action sets and probability distributions of all players except player $i$.
The expected payoff for player $ i $ is
\begin{align}
    \label{eq:payoff-mixed}
    \bar{r}^{(i)}  ( z)= \sum_{\mathbf{j}} p_{\mathbf{j}}   r^{(i)}\left( \mathbf{U}^{(s)}_\mathbf{j} \right),
\end{align}
where $p_{\mathbf{j}} = \Pi_i p^{(i)}_{j_i} $ is the joint classical probability of the strategy profile.
We rewrite the expected payoff of player $i$ as
$
r^{(i)}  =  \sum_{j_i} p^{(i)}_{j_i} \ell^{(i)}_{j_i},
$
where $\ell^{(i)}_{j_i}$
is the per-action payoff and defined as the expected payoff obtained by player~$i$ when using the pure strategy $\mathbf{U}^{(i)}_{j_i}$ while the other players follow their current mixed strategies,
\begin{align}
    \ell^{(i)}_{j_i} 
    &= \sum_{\mathbf{j}_{-i}}  p_{\mathbf{j}_{-i}} r^{(i)}\left( \mathbf{U}^{(s)}_{j_i;\mathbf{j}_{-i}} \right).
\end{align}
The vector of per-action payoffs is denoted by $\boldsymbol{\ell}^{(i)} = (\ell^{(i)}_{1}, \cdots, \ell^{(i)}_{m_i})$.

\begin{remark}
\label{remark:phaseandbounded}
The expected payoff satisfies the 
following properties:
\begin{enumerate}
\item 
\textbf{Phase invariance:} The expected payoff is invariant under global phase shifts. 
\item 
 \textbf{Boundedness:} The payoff is bounded both from below and above by the extremal eigenvalues of the payoff operator $\mathbf{R}_i$,
\begin{align}
\lambda_{\min}(\mathbf{R}_i)  \le  r^{(i)}  \le  \lambda_{\max}(\mathbf{R}_i),
\big|r^{(i)}\big|  \le  \|\mathbf{R}_i\|_{\mathrm{op}},
\end{align}
where $\lambda_{\min}(\mathbf{R}_i)$ and $\lambda_{\max}(\mathbf{R}_i)$ are the smallest and largest eigenvalues of $\mathbf{R}_i$~\cite{watrous2018theory}.
These bounds also hold for the expected payoff under a mixed strategy.
\end{enumerate}
\end{remark}

The payoff operators and density matrices are defined on compact spaces. The payoff functions are formulated over a compact strategy space consisting of unitary matrices. This strategy space is composed of differentiable operators, such as unitary transformations and inner products, which ensure finite outputs, smooth gradients, and controlled variations.
Thanks to these well-defined mathematical conditions, we can systematically study equilibria and strategic behavior in quantum games  more easily.

\textbf{Best Response (BR):}
When other players’ strategies $z^{(-i)}$, player $i$’s best response set is
\begin{align}
\label{eq:br-def}
\text{BR}^{(i)}( z^{(-i)})
=\mathop {\arg \max }\limits_{z^{(i)}}   
\bar{r}^{(i)}
\left( z^{(i)};z^{(-i)}
\right).
\end{align}
Any $z^{(i)*}\in \text{BR}^{(i)}(z^{(-i)})$ is a best response to $z^{(-i)}$ which maximizes player $i$’s expected payoff while the other players use $z^{(-i)}$.

\textbf{Nash Equilibrium:}
A joint mixed strategy  $z^*$ is a Nash equilibrium if for every player $i$,
\begin{align}
z^{(i)*} 
=\mathop {\arg \max }\limits_{z^{(i)}}   
\bar{r}^{(i)}
\left( z^{(i)};z^{(-i)*}
\right).
\end{align}
which means no player can improve their expected payoff by unilaterally changing  $ z^{(i)*}$.

\begin{theorem}
\label{thm:EEWL-Glicksberg}
In the mixed strategy EEWL game defined in Definition~\ref{def:eewl-game}, 
there exists at least one mixed strategy Nash equilibrium.
\end{theorem}

\begin{remark}
If each player is restricted to a finite set of fixed set of 
non-equivalent 
actions,  
the game reduces to a finite game, and the existence of a mixed strategy Nash equilibrium 
follows directly from Nash’s classical theorem~\cite{nash1950equilibrium}.
\end{remark}

\begin{remark}
\label{remark:nonashpure}
If the set of actions contains only a single action for  each player ($m_i = 1$), the player can still adjust their unitary strategy to search for an optimal solution.  
However, in this case, Theorem~\ref{thm:EEWL-Glicksberg} no longer applies, and the existence of a Nash equilibrium is no longer guaranteed.
\end{remark}

\section{USMEA Algorithm}
\label{sec:usmea}
In this section, we first define the decision manifold and the loss function.  
Next, we describe our approach to optimizing player strategies.  
Finally, we introduce the USMEA algorithm, which minimizes these loss functions.

The strategies live on the product manifold as the decision manifold in the EEWL quantum game,
\begin{align}
\mathcal{M}
 = 
\underbrace{\prod_{i=1}^N \prod_{j_i=1}^{m_i} \mathcal{U}(d_i)}_{\mathcal{M}_\mathcal{U}}
 \times 
\underbrace{\prod_{i=1}^N \Delta^{m_i-1}}_{\mathcal{P}},
\end{align}
with two geometries:
(i) a Riemannian geometry on the unitary manifold $\mathcal{M}_\mathcal{U}$,
(ii) a simplex geometry on the probability simplex $\mathcal{P}$.
Since both  $\mathcal{M}_{\mathcal{U}}$ and $\mathcal{P}$ are compact, their product space $\mathcal{M}$ is also compact.
Thus $\mathcal{M}$ is compact.

The payoff function in EEWL quantum games is non-convex, and multiple fixed points may exist.  
Consequently, standard global guarantees from convex optimization do not apply.  
The obtained solutions can depend on both the initialization and the chosen regularization parameters.

For optimization, we use a loss function with an entropy regularization term to ensure a proper balance between exploration and exploitation during learning.
The learning objective is to minimize the entropy-regularized loss,
\begin{align}
     L ^{(i)}(z)  =  - \bar r^{(i)} (z) -  T  H(\mathbf{p}^{(i)}), 
\end{align}
where $H(\mathbf{p}^{(i)}) = -\sum_{j_i} p^{(i)}_{j_i}\log p^{(i)}_{j_i}$ is the Shannon entropy,  
and $T>0$ is the temperature parameter.  
From a game-theoretic perspective, $T$ acts as a bounded-rationality or decision-noise parameter. It controls the trade-off between exploration and exploitation in the mixed-strategy update. When $T$ is large, the player's response becomes smoother and more exploratory. In this regime, the player is less sensitive to payoff differences, so its strategy appears more random. When $T$ is small, the update becomes more sharply concentrated on higher-payoff actions and is closer to a best response.

In USMEA, each player first updates their set of actions and then updates their mixed strategy using a softmax function applied to the current payoffs. 
For temperature $T > 0$ and step size $\eta_i > 0$, the USMEA block map $\mathcal{T}_i=\mathcal{T}_i^p\circ \mathcal{T}_i^U$ updates  $z^{(i)}$ as     
\begin{subequations}
\label{eq:USMEA-update-rules-GS}
\begin{align}
\label{eq:usmeaA-GS}
&\mathcal{T}_i^U:\quad
\mathbf{U}^{(i)[t+1]}_{j_i}
= 
\exp\big(\eta_i \mathbf{G}^{(i)[t]}_{j_i}\big) \mathbf{U}^{(i)[t]}_{j_i}, \quad \forall j_i, \\
\label{eq:usmeaB-GS}
&\mathcal{T}_i^p:\quad
\mathbf{p}^{(i)[t+1]}
= 
\sigma_T(\boldsymbol{\ell}^{(i)[t+1]}),
\end{align}
\end{subequations}
where $\mathbf{G}^{(i)}_{j_i}$ is the Riemannian gradient of the expected payoff for player $i$ over action $j_i$ is
\begin{align}
\label{eq:mixed strategy-grad-GS}
\mathbf{G}^{(i)}_{j_i}
&= 
\sum_{\mathbf{j}_{-i}}
p^{}_{(j_i;\mathbf{j}_{-i})} 
\Tr_{-i} \left(
\big[\mathbf{R}_i, 
\mathbf{U}^{(s)}_{j_i;\mathbf{j}_{-i}} 
\rho_0 
\mathbf{U}^{(s)\dag}_{j_i;\mathbf{j}_{-i}}
\big]
\right).
\end{align}
In Eq.~\eqref{eq:usmeaA-GS}, the Riemannian gradient is skew-Hermitian,  
$ \mathbf{G}^{(i)[t]}_{j_i} = -\mathbf{G}^{(i)[t]\dagger}_{j_i} $.  
Therefore, $\exp\big(\eta_i \mathbf{G}^{(i)[t]}_{j_i}\big)$ is unitary.
Since the product of two unitary matrices is also unitary, the updated action remains in $\mathcal{U}(d_i)$.
The Eq.~\eqref{eq:usmeaB-GS} is the softmax update with  $p^{(i)[t+1]}_{j_i} \ge 0$ and $\sum_{j_i} p^{(i)[t+1]}_{j_i} = 1$. Therefore, the action-update map $\mathcal{T}_i^U$ preserves the unitary manifold, and the probability-update map $\mathcal{T}_i^p$ preserves the simplex. 
Thus, $\mathcal{T}_i : \mathcal{M} \to \mathcal{M}$ is a well-defined self-map.
As a result, the algorithm requires no additional constraints to preserve the unitarity of actions or the normalization of probabilities.

The algorithm can be used to compute both BRs and fixed-point equilibrium.  
A BR iteration updates $z^{(i)}$ and fixed $z^{(-i)}$.
To compute the fixed-point equilibrium, the algorithm follows a Gauss–Seidel update scheme, where players are updated sequentially within each iteration. 
It is shown that the 
sequential update is often more locally stable than simultaneous updates of all players at once~\cite{tseng2001convergence,grippo2000convergence}.
In the sequential updates, at iteration $t$ and intermediate step $i$ we have
$z^{[t,i]} = \mathcal{T}_i ( z^{[t,i-1]})
$,
with $z^{[t+1,0]} = z^{[t,N]}$, where $z^{[t+1,0]}$ represents the beginning of iteration $t+1$.  
Equivalently, the composition
$\mathcal{T} = \mathcal{T}_N \circ \cdots \circ \mathcal{T}_1$
is called the iteration map or the one-sweep map. One iteration satisfies
$z^{[t+1,0]} = \mathcal{T}\bigl(z^{[t,0]}\bigr)$.

In a standard Riemannian gradient descent combined with softmax (RGD+Softmax), all variables are updated at once. By contrast, USMEA uses a player-wise sequential block-update scheme, in which the players are updated in Gauss–Seidel order. This distinction matters in the coupled multi-agent EEWL game, because each player tries to minimize its own loss, while updating one player’s block can also change the losses of the others through the coupled dynamics.

We present the procedure for finding 
the best response $z^{(i)*}$ in Algorithm~\ref{alg:Hybrid-USMEU-BR} and the fixed point $z^{*}$ in Algorithm~\ref{alg:Hybrid-USMEU}.
\begin{algorithm}[b]
\caption{\label{alg:Hybrid-USMEU-BR}
USMEA  algorithm to find BR}
\begin{algorithmic}[1]
\Require $\eta_i$: learning parameter, $\epsilon$: convergence tolerance, $i$: player index, 
$T$: initial temperature for classical probabilities,
$\mathbf{U}^{-i}_{j_{-i}}$: Set of all players' action except player $i$
\State Initialize a random set of actions $\{\mathbf{U}^{(i)[0]}_1,\mathbf{U}^{(i)[0]}_2, \cdots, \mathbf{U}^{(i)[0]}_{m_i}\}$  and corresponding probabilities  $\{p^{(i)[0]}_{1},p^{(i)[0]}_{2}, ..., p^{(i)[0]}_{m_i}\}$ for player $i$
\Repeat{($k=0,1,\cdots$: each step)}
\State Calculate $\{\mathbf{G}^{(i)[k]}_{j_i}\}_{j_i=1}^{m_i}$
\For {every  $j_i \in \{1,2,\cdots, m_i\}$}
\State $\mathbf{U}^{(i)[k+1]}_{j_i} \gets 
e^{\eta_i \mathbf{G}^{(i)[k]}_{j_i}} \mathbf{U}^{(i)[k]}_{j_i}$
\EndFor
\State Calculate $\boldsymbol{\ell}^{(i)[k+1]}$
\State $\mathbf{p}^{(i)[k+1]} \gets \sigma_T(\boldsymbol{\ell}^{(i)[k+1]})$
\State  Update $T$
\Until{$\left\|\mathbf{U}^{(i)[k+1]}_{j_i}-\mathbf{U}^{(i)[k]}_{j_i}\right\|\le \epsilon$ 
and 
$\left\|p^{(i)[k+1]}_{j_i}-p^{(i)[k]}_{j_i}\right\|\le \epsilon$ $j_i$}
\end{algorithmic}
\end{algorithm}
\begin{algorithm}[h]
\caption{\label{alg:Hybrid-USMEU}
USMEU algorithm to find fixed points}
\begin{algorithmic}[1]
\Require $\{\eta_i\}$: learning parameter for USMEU, $\epsilon$: convergence tolerance,
$T$: initial temperature for classical probabilities,
$\epsilon$: convergence tolerance
\State Initialize a random set of quantum actions $\{\mathbf{U}^{(i)[0]}_1,\mathbf{U}^{(i)[0]}_2, \cdots, \mathbf{U}^{(i)[0]}_{m_i}\}$  and corresponding probabilities  $\{p^{(i)[0]}_{1},p^{(i)[0]}_{2}, ..., p^{(i)[0]}_{m_i}\}$ for each player $i \in \mathcal{N}$ 
\Repeat{($k=0,1,\cdots$: each step)}
\For{every player $i \in \mathcal{N}$}
\State Calculate $\{\mathbf{G}^{(i)[k]}_{j_i}\}_{j_i=1}^{m_i}$ according to the GS update rule.
\For {every  $j_i \in \{1,2,\cdots, m_i\}$}
\State $\mathbf{U}^{(i)[k+1]}_{j_i} \gets e^{\eta_i \mathbf{G}^{(i)[k]}_{j_i}} \mathbf{U}^{(i)[k]}_{j_i} $
\EndFor
\State Calculate $\boldsymbol{\ell}^{(i)[k+1]}$
\State $\mathbf{p}^{(i)[k+1]} \gets \sigma_T(\boldsymbol{\ell}^{(i)[k+1]})$
\EndFor
\State Update $T$
\Until{$\left\|\mathbf{U}^{(i)[k+1]}_{j_i}-\mathbf{U}^{(i)[k]}_{j_i}\right\|\le \epsilon$ 
and 
$\left\|p^{(i)[k+1]}_{j_i}-p^{(i)[k]}_{j_i}\right\|\le \epsilon$
for all $i,j_i$}
\end{algorithmic}
\end{algorithm}

\section{Theoretical Analysis}
\label{sec:theoretical}
In this section, we present the theoretical analysis of the proposed USMEA algorithm.  
First, we derive the Riemannian gradient from the payoff function.  
We then compute the Lipschitz coefficients and identify a safe interval for the learning rate.
Finally, we prove the results for the BR updates and extend them to the full sequential setting to obtain fixed-point equilibria of the joint dynamics.

Each unitary action $j_i$ is updated on the Riemannian manifold using the exponential retraction defined in Eq.~\eqref{eq:usmeaA-GS}.  
The term 
$\mathbf{G}^{(i)}_{j_i} = \mathrm{grad}_{\mathbf{U}^{(i)}_{j_i}}  L ^{(i)}$ represents the Riemannian gradient on the unitary group manifold~\cite{absil2008optimization}.  
\begin{theorem}
\label{theorem:calculateG}
Let $\mathbf{G}^{(i)}_{j_i}$ be the Riemannian gradient associated with the unitary updates.  
Under the USMEA update rules, we have 
\begin{enumerate}
\item  $
   \mathrm{grad}_{\mathbf{U}^{(i)}_{j_i}}  L ^{(i)} 
    = - \mathrm{grad}_{\mathbf{U}^{(i)}_{j_i}} \bar{r}^{(i)}.$
\item The explicit form of the Riemannian gradient is given by Eq.~\eqref{eq:mixed strategy-grad-GS}.
\end{enumerate}
\end{theorem}

The payoff operator $\mathbf{R}_i$ is bounded, $\|\mathbf{R}_i\|_{\mathrm{op}} < \infty$. The map 
$
\rho \mapsto \bar{r}^{(i)}
$
is smooth over the manifold $\mathcal{M}$. 
Furthermore, since $\mathbf{R}_i$, $\rho_0$, and the  unitary operators $\mathbf{U}^{(i)}_{j_i}$ are analytic, the expected payoff $\bar{r}^{(i)}$ is itself an analytical function.
Since $ \mathcal{M} $ is compact and $ \bar r^{(i)} $ is analytic, its derivative is bounded on $ \mathcal{M} $.
Therefore $ \bar r^{(i)} $ is Lipschitz continuous on the manifold $ \mathcal{M} $.
Estimating the Lipschitz coefficients of the payoff gradients on $ \mathcal{M} $ lets us choose safe learning rates and bound the per-iteration change in $  L ^{(i)} $. 
To begin with, we define the distance between two action sets on $\mathcal{M}$.  
Then, we compute the corresponding Lipschitz coefficients.
\begin{definition}
The geodesic distance between unitary matrices and the distance between probability values are defined as follows
\begin{enumerate}
\item For player $i$, the block (geodesic) distance between two unitary matrices $\mathbf{U}^{(i)}_{j_i}$ and $\mathbf{U}'^{(i)}_{j_i}$ is ~\cite{neff2014logarithmic}
\begin{align}
\label{eq:distanceunitary}
d_{\mathcal{U}} \big(\mathbf{U}^{(i)}_{j_i}, \mathbf{U}'^{(i)}_{j_i}\big)
&= \left\| \log \left( \mathbf{U}^{(i)\dagger}_{j_i} \mathbf{U}'^{(i)}_{j_i} \right) \right\|_F.
\end{align}
\item For player $i$, the distance between two unitary sets $\{\mathbf{U}_{j_i}^{(i)}\}_{j_i=1}^{m_i}$ and 
$\{\mathbf{U}'^{(i)}_{j_i}\}_{j_i=1}^{m_i}$, and the probability vectors $\mathbf{p}^{(i)}$ and $\mathbf{p}'^{(i)}$ is
\begin{subequations}
\begin{align}
\Delta_{\mathcal{U}}^{(i)} \left( \{\mathbf{U}_{j_i}^{(i)}\}, \{\mathbf{U}'^{(i)}_{j_i}\} \right)
&= \sum_{j_i=1}^{m_i} d_{\mathcal{U}} \left(\mathbf{U}^{(i)}_{j_i}, \mathbf{U}'^{(i)}_{j_i}\right), \\[4pt]
\Delta_p^{(i)}\left( \mathbf{p}^{(i)}, \mathbf{p}'^{(i)} \right)
&= \left\| \mathbf{p}^{(i)} - \mathbf{p}'^{(i)} \right\|_2.
\end{align}
\end{subequations}
The total distance for the player $i$ is then defined by combining these via the Euclidean norm
\begin{align}
\Delta^{(i)} = \left\| \left( \Delta_{\mathcal{U}}^{(i)},  \Delta_p^{(i)} \right) \right\| 
= \sqrt{ \left( \Delta_{\mathcal{U}}^{(i)} \right)^2 + \left( \Delta_p^{(i)} \right)^2 }.
\end{align}
\item When all players update their actions and probability distributions from $z$ to $z'$, the distances are 
\begin{subequations}
\begin{align}
\Delta_{\mathcal{U}}^2
&= \sum_{i=1}^N  \Delta_{\mathcal{U}}^{(i)} \left( \{\mathbf{U}_{j_i}^{(i)}\}_{j_i=1}^{m_i},\{\mathbf{U}'^{(i)}_{j_i}\}_{j_i=1}^{m_i} \right), \\
\Delta_p^2 
&=\sum_{i=1}^N  \left\| \mathbf{p}^{(i)} - \mathbf{p}'^{(i)} \right\|^2_2.
\end{align}
\end{subequations}
The total joint distance is then given by
\begin{align}
\Delta(z,z')
= \sqrt{\Delta_{\mathcal{U}}^2 + \Delta_p^2}.
\end{align}
\end{enumerate}
\end{definition}

\begin{theorem}
\label{thm:block-lip}
Suppose that the player $i$ modifies only the unitary action $j_i$,
then
\begin{enumerate}
\item 
The gradient $\mathbf{G}^{(i)}_{j_i}$ is Lipschitz continuous along the unitary geodesic with constant $A_i$ as
\begin{align}
\left\|\mathbf{G}^{(i)}_{j_i}
(\mathbf{U}^{(i)}_{j_i}) - \mathbf{G}^{(i)}_{j_i}(\mathbf{U}'^{(i)}_{j_i})\right\|_F
 \le 
A_i  d_{\mathcal{U}} \left(\mathbf{U}^{(i)}_{j_i}, \mathbf{U}'^{(i)}_{j_i}\right),
\end{align}
where
\begin{align}
A_i = 4 \sqrt{ d_{-i}} \|\mathbf{R}_i\|_{\mathrm{op}} \|\rho_0\|_F.
\end{align}
\item 
The per-action payoff $\ell^{(i)}_{j_i}$ is Lipschitz continuous with constant $M_i$ as
\begin{align}
\label{eq:lLipschitz}
\left| \ell^{(i)}_{j_i}(\mathbf{U}^{(i)}_{j_i}) - \ell^{(i)}_{j_i}(\mathbf{U}'^{(i)}_{j_i}) \right|
 \le 
M_i  d_{\mathcal{U}}  \left(\mathbf{U}^{(i)}_{j_i}, \mathbf{U}'^{(i)}_{j_i}\right),
\end{align}
where
\begin{align}
M_i = 2 \sqrt{ d_{-i}} \|\mathbf{R}_i\|_{\mathrm{op}} \|\rho_0\|_F.
\end{align}
\end{enumerate}
\end{theorem}

When the actions are updated according to USMEA, Eq.~\eqref{eq:usmeaA-GS},  
the previous probability distribution no longer minimizes the loss function.  
Instead, the new minimizer of the loss is given by the softmax distribution in Eq.~\eqref{eq:usmeaB-GS}.

\begin{theorem} 
\label{thm:bundle-descent-usmea}
Fix
$z^{(-i)}$.
$z^{(i)}=(\{\mathbf{U}^{(i)}_{j_i}\}, \mathbf{p}^{(i)})$ to $z^{(i)'}=(\{\mathbf{U}'^{(i)}_{j_i}\}, \mathbf{p}'^{(i)})$
using USMEA.
If the step size satisfies $\eta_i \le 1 / A_i$, then the update is monotonic in $ L ^{(i)}$, and we have the following guaranteed improvement
\begin{align}
\label{eq:bundle-descent-usmea}
 L ^{(i)}\left(z^{(i)'};z^{(-i)}\right)
& \le
 L ^{(i)}\left(z^{(i)};z^{(-i)}\right)
\nonumber \\
&-
\eta_i\left(1 - \frac{1}{2} \eta_i A_i \right) 
\Delta_\mathcal{U}^{(i)2}
\nonumber \\
&-
\frac{T}{2} \Delta^{(i)2}_p %
\end{align}
In particular, the right-hand side is greater unless simultaneously $\mathbf{G}^{(i)}_{j_i} = 0$ for all $j_i$ and $\mathbf{p}'^{(i)} = \mathbf{p}^{(i)}$.
\end{theorem}

\begin{remark}
The bound in Eq.~\eqref{eq:bundle-descent-usmea} is fully explicit in terms of $(\eta_i, T)$ and the defined block constants.  
It shows that the USMEA update, when $z^{(-i)}$, constitutes a descent step for $L^{(i)}$.  
In the limit $T \to 0$, the softmax gain vanishes, and the update reduces to the payoff-based ascent term.
\end{remark}
\begin{theorem}
\label{theorem:BR-USMEA}
Fix $z^{(-i)}$, and assume that $\eta_i \in (0,1/A_i]$.  
Then the USMEA updates for player $i$ converge to a fixed point 
$z^{(i)*}=(\mathbf{U}^{(i)*}_{j_i}, \mathbf{p}^{(i)*})$.  
At this point, for every $j_i$ we have
\begin{align}
   \mathbf{G}^{(i)}_{j_i}(z^{(i)*};z^{(-i)}) = 0, \quad
   \mathbf{p}^{(i)*}
   = \sigma_T(\boldsymbol{\ell}^{(i)*}),
\end{align}
where $\ell^{(i)*}_{s_i}$ denotes the per-action expected payoff at this fixed point. 
\end{theorem}
\begin{remark}
The step-size condition $\eta_i \in (0,1/A_i]$ is a guaranteed safe interval for the USMEA BR update. Since $A_i$ is derived from a worst-case Lipschitz bound on the Riemannian gradient, this interval is not necessarily sharp. In practice, larger step sizes may still work empirically.
\end{remark}

\begin{remark}
The $z^{(i)\star}$ is a BR-stationary point of $L^{(i)}$. It is not always locally attracting. If $z^{(i)\star}$ is a strict local minimizer of $L^{(i)}$, then it is locally asymptotically stable for gradient flow and locally attracting for the BR update when the step size is sufficiently small. If it corresponds to a nondegenerate saddle of $L^{(i)}$, then it is not locally attracting. As $T \to 0$, an attracting BR fixed point corresponds to a local best response to $z^{(-i)}$ \cite{absil2006stable,lee2016gd,mckelvey1995quantal}.
\end{remark}

\begin{remark}
\label{remark:Lconvergencenotenough}
Convergence of $ L ^{(i)}$ does not necessarily imply convergence of the expected payoff $\bar r^{(i)}$.  
Indeed, it is possible for the entropy term $H^{[t]}$ and the expected payoff sequence $\bar r^{(i)[t]}$ to vary in such a way that their combination keeps $ L ^{(i)}$ constant.
Therefore, convergence of $L^{(i)}$ alone is not a reliable stopping criterion for the BR update.
\end{remark}
\begin{corollary}
\label{cor:oscillation}
Fix $z^{(-i)}$.  
Suppose $ L ^{(i)[t]}$ converges to a fixed point.  
Then the expected payoff sequence $\{\bar r^{(i)[t]}\}$ satisfies the oscillation bound,
\begin{align}
\limsup_{t} \bar r^{(i)[t]} - \liminf_{t} \bar r^{(i)[t]}  \le  T \log m_i.
\end{align}
\end{corollary}
By Corollary~\ref{cor:oscillation}, the expected payoff sequence $\bar{r}^{(i)}$ may oscillate within a band of width at most $T \log m_i$.  
Based on Remark~\ref{remark:Lconvergencenotenough} and Corollary~\ref{cor:oscillation}, the convergence of $L^{(i)}$ does not imply the convergence of $\bar{r}^{(i)}$ for $T > 0$.  
Therefore, it is not advisable to use the convergence of $L^{(i)}$ as the sole stopping criterion in the algorithm. 
As $T \to 0$, this oscillation window shrinks to zero.  
Consequently, when $z^{(-i)}$ is fixed, the sequence $\{\bar{r}^{(i)}_t\}$ converges to the payoff of a (local) BR fixed point for player $i$.
Now, we turn to the joint dynamics in which all players’ actions and mixing strategies are updated sequentially.

\begin{theorem}
\label{theorem:exist-fp}
Consider all players update their actions and probabilities using
iteration map $\mathcal{T} = \mathcal{T}_N \circ \cdots \circ \mathcal{T}_1$.  
Then there exists at least one fixed point $z^\star \in \mathcal{M}$ such that $\mathcal{T}(z^\star) = z^\star$. In other words, 
\begin{subequations}
\begin{align}
 &\mathbf{G}^{(i)}_{j_i}(z^\star)= 0,\quad
 p^{(i)\star}_{j_i}
 = \sigma_T(\boldsymbol{\ell}^{(i)*})_{j_i}
, \quad \forall j_i, i.
\end{align}
\end{subequations}
\end{theorem}

\begin{remark}
With a single action per player ($m_i=1$), the entropy term in the loss function vanishes and USMEA reduces to unitary updates.  
In this case, a Nash equilibrium may not exist (Remark~\ref{remark:nonashpure}), but fixed points still do.  
To verify whether a fixed point is a Nash equilibrium, one can compute the BR at that point.
\end{remark}
\begin{remark}
In the mixed strategy setting ($m_i>1$), a fixed point at temperature $T>0$ is a quantal–response (logit) equilibrium. 
For a fixed point with at least two non-equivalent actions per player, as the temperature is annealed to zero, the corresponding equilibrium path converges to a Nash equilibrium~\cite{mckelvey1995quantal,balter2024decision}.
\end{remark}

When a player updates $z^{(i)}$ to to improve its loss, the change influences not only its own loss but also the losses of other players. 
A single block update affects the entire game.  
As a result, convergence should be studied collectively through the full Gauss–Seidel one-sweep map,
viewed as a discrete dynamical system on $\mathcal{M}$~\cite{NippStoffer1992}. The relevant local object is the differential
$D\mathcal{T}(z^\star):T{z^\star}\mathcal{M}\to T{z^\star}\mathcal{M}$
at a fixed point $z^\star$. Since fixed points of $\mathcal{T}$ need not be isolated and may form a $C^1$ embedded submanifold of $\mathcal{M}$, local convergence requires conditions that permit neutral directions tangent to the fixed-point set and contraction in the transverse directions.

\begin{theorem}[Local convergence of USMEA near a non-isolated fixed-point]
\label{theorem:convergece_all_local}
Let $\mathcal M^{(\mathrm{FP})} \subset \mathcal M$ be a $C^1$ embedded submanifold of fixed points of $\mathcal T$, that is,
$
\mathcal T(z)=z, \forall z \in \mathcal M^{(\mathrm{FP})}.
$
Fix $z^\star \in \mathcal M^{(\mathrm{FP})}$.
Assume that:
\begin{enumerate}
    \item $\dim \mathcal M^{(\mathrm{FP})}=l$.
    \item The eigenvalue $1$ of $D\mathcal T(z^\star)$ has geometric multiplicity $l$, and the eigenvalue $1$ is semisimple.
    \item Every other eigenvalue $\lambda$ of $D\mathcal T(z^\star)$ satisfies
    $
    |\lambda|<1.
    $
\end{enumerate}
Then there exists a neighborhood $\mathcal V \subset \mathcal M$ of $z^\star$ such that, for every $z^{[0]} \in \mathcal V$, the iterates
$
    z^{[k+1]}=\mathcal T(z^{[k]})
$
remain in $\mathcal V$ and converge to a limit
$
    z^{\infty} \in \mathcal M^{(\mathrm{FP})}.
$
In particular, the USMEA dynamics converge locally to a fixed point in the embedded fixed-point submanifold $\mathcal M^{(\mathrm{FP})}$.
\end{theorem}
\begin{remark}[Isolated fixed point]
When the fixed point $z^\star$ is isolated, the fixed-point submanifold $\mathcal M^{(\mathrm{FP})}$ is locally zero-dimensional, that is, $l=0$. In this case, there are no tangent neutral directions, and the eigenvalue $1$ does not appear in $D\mathcal T(z^\star)$. Therefore, the assumptions of Theorem~\ref{theorem:convergece_all_local} reduce to the standard local asymptotic stability condition that every eigenvalue $\lambda$ of $D\mathcal T(z^\star)$ satisfies $|\lambda|<1$. Equivalently, the spectral radius of $D\mathcal T(z^\star)$ is strictly smaller than $1$.
\end{remark}

To compute the differential of the one-sweep map at a fixed point $z^\star$, we use the chain rule on smooth manifolds, as
\begin{align}
	D\mathcal T(z^\star)
	=
	D\mathcal T_N(z^\star)\cdots D\mathcal T_1(z^\star).
\end{align}
Similarly, for each block update $\mathcal T_i$, we have
$
	D\mathcal T_i(z^\star)
	=
	D\mathcal T_i^{p}(z^\star) D\mathcal T_i^{U}(z^\star).
$
Therefore, the computation of $D\mathcal T(z^\star)$ reduces to computing $D\mathcal T_i^{U}(z^\star)$ and $D\mathcal T_i^{p}(z^\star)$ for all players, forming each block differential $D\mathcal T_i(z^\star)$, and multiplying them in the Gauss-Seidel order. 
The eigenvalues used in the local convergence analysis are then obtained from the resulting linear map $D\mathcal T(z^\star)$.
The detailed calculations of $D\mathcal T_i^{U}(z^\star)$ and $D\mathcal T_i^{p}(z^\star)$ are provided in Section~\ref{supp:DT} of the Supplementary Material.

\begin{remark}
For any permutation $\kappa$ of $\{1,2,\ldots,N\}$, define 
$
\mathcal T_\kappa  =\; \mathcal T_{\kappa(N)} \circ \cdots \circ \mathcal T_{\kappa(1)}
$.
All theorems and lemmas in this paper apply to every $\mathcal T_\kappa$.
None of our arguments depend on the specific update order. 
However, in general, since
$
D\mathcal T_\kappa(z^\star)\neq D\mathcal T(z^\star),
$
so the corresponding eigenvalues may also differ. Therefore, the local spectral conditions in Theorem~\ref{theorem:convergece_all_local} must be verified for the chosen update order.
\end{remark}

\section{Experimental Analysis}
\label{sec:expriments}
In this section, we test the performance of our proposed algorithm in finding BRs and fixed points.  
First, to study the stability of the Riemannian component of the algorithm in the quantum Prisoner’s Dilemma, we consider the case where each player uses a single unitary operator, that is, $m_i = 1$ for all $i$.  
In this case, there is no softmax update, and only the unitary (Riemannian) step is applied.  
Next, we analyze the algorithm in the quantum Prisoner’s Dilemma with mixed strategies.  
Finally, we examine its performance in larger strategy spaces as well as in a three-player Prisoner’s Dilemma.

In the plots, we use $(\text{BR}^{(i)}; z^{\star})$ to show the BR to the fixed point $z^{\star}$  by player $i$.
To evaluate the convergence behavior of the learning algorithm, we use the relative error defined as  
\begin{align}
    \varepsilon_{\text{rel}} =  \frac{|\overline{r}^{(i)}(z^{[k]}) - \bar{r}^{(i)}(z^{\star})|}{\bar{r}^{(i)}(z^{\star})}.
\end{align}
We use $\bar{\varepsilon}_{\text{rel}}$ 
to denote the average relative error
over $n$ configurations initialized with different random seeds.
The annealing temperature is adjusted according to $T = T_0 \alpha^{k}$, where $T_0$ is the initial temperature.

Since each unitary action is not unique, 
we only present the corresponding payoffs and BRs.
To determine whether a fixed point is a Nash equilibrium, we compute the BR for all players at that fixed point. If any player can adopt an alternative strategy that increases their payoff, the fixed point cannot be classified as a Nash equilibrium.

\begin{table}
\centering
\renewcommand{\arraystretch}{1.2}   
\normalsize      

\begin{subtable}[t]{\columnwidth}  
\centering
\begin{tabular}{c|c||c|c}
\textbf{Profile} & \textbf{Payoffs} & \textbf{Profile} & \textbf{Payoffs} \\ \hline
(C,C) & (3,3) &  (D,C) & (5,0)\\
(C,D) & (0,5) & (D,D) & (1,1) \\
\end{tabular}
\caption{Two-player Prisoner’s Dilemma.}
\label{tab:dilemma}
\end{subtable}

\vspace{0.5em} 

\begin{subtable}[t]{\columnwidth}
\centering
\begin{tabular}{c|c||c|c}
\textbf{Profile} & \textbf{Payoffs} & \textbf{Profile} & \textbf{Payoffs} \\ \hline
(C,C,C) & (3,3,3) & (D,C,C) & (5,2,2) \\
(C,C,D) & (2,2,5) & (D,C,D) & (4,0,4) \\
(C,D,C) & (2,5,2) & (D,D,C) & (4,4,0) \\
(C,D,D) & (0,4,4) & (D,D,D) & (1,1,1) \\
\end{tabular}
\caption{Three-player Prisoner’s Dilemma.}
\label{tab:dilemma3p}
\end{subtable}
\caption{Payoff table for classical Prisoner’s Dilemma~\cite{du2002playing}.}
\label{tab:pd-side-by-side}
\end{table}
\begin{figure}
    \centering
    \includegraphics[width=0.99\linewidth]{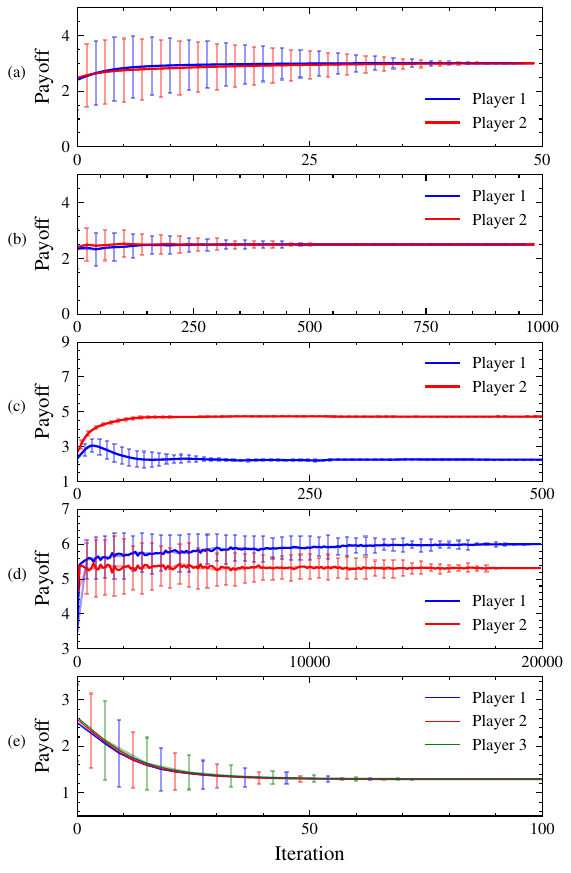}
\caption{
Convergence of the expected payoffs with $n = 200$ and $\eta = 0.05$.  
(a) Game 1: Prisoner’s Dilemma with pure strategies and $\gamma = \pi/2$.  
(b) Game 1: Prisoner’s Dilemma with mixed strategies, with $T_0 = 10$, $\alpha = 0.99995$, and $\gamma = \pi/2$.  
(c) Game 2, with $T_0 = 1$, $\alpha = 0.9995$, and $\gamma = \pi/2$.  
(d) Game 3, with $T_0 = 1$, $\alpha = 0.99995$, and $\gamma = \pi/2$.  
(e) Game 4, with $\gamma = \pi/8$.
}
    \label{fig:convergence}
\end{figure}

In these experiments, we consider four different quantum games as follows
\begin{enumerate}
\item \textbf{Game 1:} Two-player quantum Prisoner’s Dilemma.
\item \textbf{Game 2:} Two-player quantum game with asymmetric strategy spaces ($d_1 = 2$ and $d_2 = 3$).
\item \textbf{Game 3:} Two-player quantum game with $d_1 = d_2 = 3$.
\item \textbf{Game 4:} Three-player quantum Prisoner’s Dilemma.
\end{enumerate}
For each quantum game, the details of the initial states, outcome sets, and the corresponding outcome payoffs are provided in Supplementary Material, Section~\ref{append:exp_info}. 
We use $\gamma \in [0,\pi/2]$ as the entanglement parameter.
When $\gamma=0$, the state is unentangled (separable) and the outcomes reduce to the classical game.
When $\gamma=\pi/2$, the state is maximally entangled~\cite{eisert1999quantum}.
In this scenario, separable initial states regain classical strategic behavior, while entangled initial states introduce quantum correlations that change the equilibrium structure.

The payoff tables for the classical two-player and three-player Prisoner’s Dilemma are shown in Table~\ref{tab:pd-side-by-side}.  
In the classical Prisoner’s Dilemma, the Nash equilibrium occurs when all players choose to defect, which gives them a payoff of 1.

In Fig.~\ref{fig:convergence}, the trajectories of the expected payoffs for Games 1-4 are shown as they converge to their respective fixed points. 
As the strategy space expands and the number of players increases, the algorithm converges more slowly.

\textbf{Nash equilibrium and BRs in the quantum Prisoner’s Dilemma:}\\
\begin{figure}[t]
    \centering
    \includegraphics[width=0.99\linewidth]{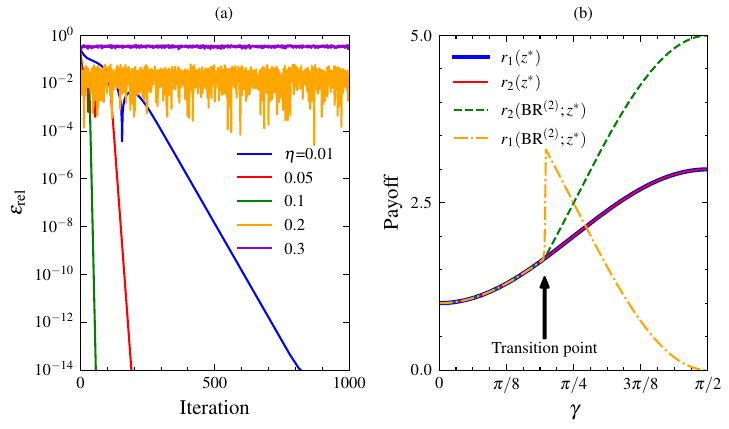}
    \caption{Quantum Prisoner's Dilemma with pure quantum strategy.
(a) Average relative payoff error with $n = 100$.
(b) Payoff at the fixed point for different $\gamma$.
The transition occurs is $\gamma_{\mathrm{TP}}=0.616\pm 0.001$.
For $\gamma<\gamma_{\mathrm{TP}}$, a pure-strategy Nash equilibrium set exists, while for $\gamma>\gamma_{\mathrm{TP}}$ no pure-strategy Nash equilibrium exists.
Beyond the transition, players can find BRs that counter their opponents’ strategies.
}
    \label{fig:quantum-PD-pure}
\end{figure}
Fig.~\ref{fig:quantum-PD-pure} shows  
the convergence of USMEA for the quantum Prisoner’s Dilemma with pure strategies.
Due to the symmetry between the two players in this game, the  payoff remains identical for both of them throughout the learning process.
In Fig.~\ref{fig:quantum-PD-pure}-(a), the convergence behavior of the payoff for different learning rates is examined.
When the learning rate is very small, the convergence is slow and requires many iterations to stabilize. As we increase the learning rate, convergence becomes faster. However, if the learning rate becomes too large, the algorithm fails to converge to the optimal solution. Instead, it may overshoot or oscillate around potential solutions, failing to reach the optimal payoff.

In Fig.~\ref{fig:quantum-PD-pure}-(b), the payoffs at the fixed point and the corresponding best-response payoffs are shown as functions of the entanglement parameter $\gamma$.
When $\gamma$ is small, the game owns a pure-strategy Nash equilibrium set.
However, as $\gamma$ increases, the players' strategies become more strongly coupled through nonclassical correlations.
We observe a transition point at $\gamma_{\text{TP}} = 0.616 \pm 0.001$, beyond which the pure-strategy Nash equilibrium disappears.
Beyond the transition point, at least one player can find a counter-strategy that increases its own payoff and lowers the opponent's payoff.
The dashed curves show the payoffs when Player~2 uses the best response to the fixed point. Beyond the transition point, Player~2 increase its own payoff and reduce Player~1's payoff.
Hence, the fixed point is no longer a mutual best response.
Physically, $\gamma_{\text{TP}}$ is the entanglement threshold at which the game enters an entanglement-dominated strategic regime. Below this threshold, pure-strategy equilibria survive, whereas above it the nonclassical correlations are strong enough to destroy the pure-strategy equilibrium.
We emphasize that this is a phase-transition-like change in equilibrium structure, not a thermodynamic phase transition~\cite{du2003phase,bugu2025entanglement}.
At the maximum entanglement level, $\gamma = \frac{\pi}{2}$, the optimal strategy leads to a cooperative payoff of $(3, 3)$, even though the BR function allows one player to obtain a payoff of $5$ and decrease the other player's payoff. 
Our experiments further reveal beyond the transition point, the total payoff remains constant at $r_1 + r_2 = 5$. 
These behaviors agree with Ref.~\cite{du2002playing} and related studies.
For more than $95\%$ of random action initializations, the algorithm converged to these fixed points in this model.

\begin{figure}[t]
    \centering
    \includegraphics[width=0.99\linewidth]{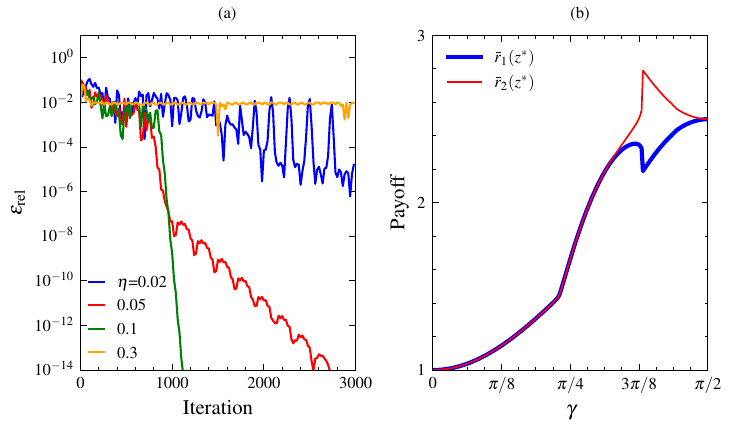}
    \caption{Quantum Prisoner's Dilemma with mixed strategies ($m_i=2$), $T_0 = 10$, and $\alpha = 0.99995$.
(a) Average relative payoff error with $n = 100$.
(b) Expected payoff at the fixed point for different $\gamma$.
}
\label{fig:quantum-PD-mixed}
\end{figure}
The convergence behavior of the algorithm for the two-player quantum Prisoner’s Dilemma with mixed strategies is shown in Fig.~\ref{fig:quantum-PD-mixed}.  
For mixed strategies and $\gamma = \pi/2$, the expected payoff converges to $2.5$ for each player. 
This point is a Nash equilibrium and yields a higher expected payoff than the classical Prisoner’s Dilemma Nash equilibrium, which is $1$.
Fig.~\ref{fig:quantum-PD-mixed}-(a) shows the convergence of the expected payoff for different learning rates $\eta$.
Fig.~\ref{fig:quantum-PD-mixed}-(b) shows the expected payoff at the fixed point for different values of $\gamma$.
We start from $\gamma = \pi/2$ and let the algorithm converge to a fixed point.
Since changing $\gamma$ continuously shifts the fixed point, we decrease $\gamma$ gradually and use the previous fixed point as the initial condition for the next run.
Because this convergence point is attracting, the algorithm can easily find the new fixed point with this approach. 
This method helps us track a specific fixed point when the system has multiple fixed points, especially when there is symmetry between different players.

\begin{table}[t]
\centering
\normalsize   
\begin{tabular}{c c c c c}
	\hline
	$m_i$ & $\gamma$ & $l$ & semisimple at $\lambda=1$ & $r_\perp$ \\
	\hline
	$1$ & $0$      & $2$ & yes & $0.9500$ \\
	$1$ & $\pi/8$  & $1$ & yes & $0.9709$ \\
	$1$ & $\pi/2$  & $3$ & yes & $0.9350$ \\
	$2$ & $0$      & $4$ & yes & $0.9750$ \\
	$2$ & $\pi/8$  & $1$ & yes & $0.9927$ \\
	$2$ & $\pi/2$  & $4$ & yes & $0.9875$ \\
	\hline
\end{tabular}
\caption{Local spectral conditions for the two-player EEWL Prisoner's Dilemma at $T=1$ and $\eta=0.05$, evaluated at selected numerically computed fixed points. Here, $r_\perp=\max_{\lambda\not\approx 1}|\lambda|$ is the largest modulus among the non-neutral eigenvalues.}
\label{table:expecturn}
\end{table}
To assess whether the observed local behavior is consistent with Theorem~\ref{theorem:convergece_all_local}, we next examine the spectrum of $D\mathcal T(z^\star)$ at a few representative fixed points computed numerically. Table~\ref{table:expecturn} reports the corresponding local spectral quantities.

To further examine the convergence of the proposed update rule, we compare USMEA with the RGD+Softmax under the same initialization. Fig.~\ref{fig:usmea-RGDSoftmax} shows payoff trajectories of player 1 in game 1 for different temperatures and step sizes. In Fig.~\ref{fig:usmea-RGDSoftmax}-(a), both methods converge for sufficiently small step sizes. For larger step sizes in the tested range, however, USMEA remains stable and converges faster, whereas RGD+Softmax exhibits stronger oscillations and may fail to converge (see Fig.~\ref{fig:usmea-RGDSoftmax}-(b) and Fig.~\ref{fig:usmea-RGDSoftmax}-(c)). 
Across many random initializations (not shown here), USMEA converged in a larger fraction of runs than RGD+Softmax.
Both methods can lose convergence when the step size is chosen outside their stable regime.  Fig.~\ref{fig:usmea-RGDSoftmax}-(b) and Fig.~\ref{fig:usmea-RGDSoftmax}-(d), which use the same step size but different temperatures, shows that the lower temperature makes the oscillatory behavior of RGD+Softmax more pronounced. Overall, these experiments suggest that USMEA has better empirical robustness to initialization, step size, and temperature in the tested mixed-strategy setting.
\begin{figure}[t]
    \centering
    \includegraphics[width=0.99\linewidth]{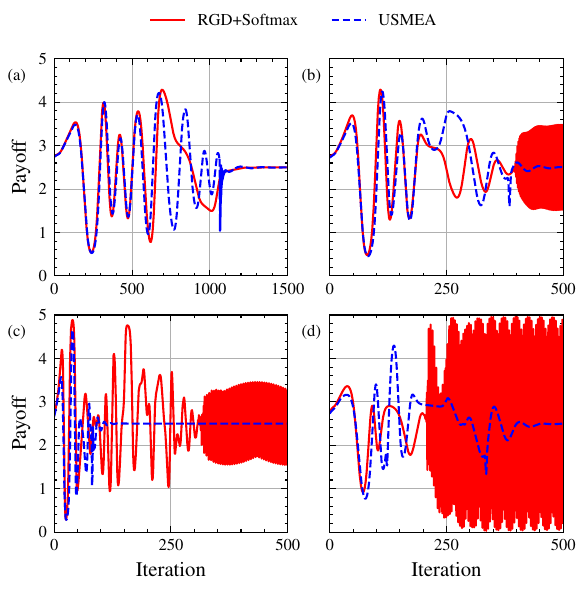}
    \caption{
    Comparison between USMEA and the RGD+Softmax baseline in the mixed-strategy Quantum Prisoner's Dilemma ($m_i=2$), using the same initialization and showing the payoff trajectory of player 1 in Game 1. 
    (a) $T=1$ and $\eta = 0.01$
    (b) $T=1$ and $\eta = 0.03$
    (c) $T=1$ and $\eta = 0.1$
    (d) $T=0.5$ and $\eta = 0.03$
    }
    \label{fig:usmea-RGDSoftmax}
\end{figure}

\textbf{
Complexity and Convergence Discussion:
}
\\
The general form of an $n \times n$ unitary matrix 
$\mathbf{U} \in \mathcal{U}(n)$
is given by
\begin{align}
    \label{eq:generalUnitary}
    \mathbf{U} = \begin{pmatrix}
        U_{11} & U_{12} & \cdots & U_{1n} \\
        U_{21} & U_{22} & \cdots & U_{2n} \\
        \vdots & \vdots & \ddots & \vdots \\
        U_{n1} & U_{n2} & \cdots & U_{nn}
    \end{pmatrix},
\end{align}
where $U_{ij} \in \mathbb{C}$. 
The unitary group $ \mathcal{U}(n)$ has $n^2$  degrees of freedom. However, ignoring the global phase invariance, there are effectively $n^2 - 1$ independent complex parameters.
In the EEWL quantum game, the dimension of the overall Hilbert space grows exponentially with the number of players and corresponding strategy space, scaling as $O\left(\prod_{i \in \mathcal{N}} d_i\right)$. However, due to the product (or structured) nature of the joint strategy operator $\mathbf{U}^{(s)}$, the number of parameters that actually need to be optimized remains much smaller. 
The number of trainable decision variables scales as 
$O\!\left(\sum_{i\in\mathcal N} m_i (d_i^2-1)\right)$, which reflects the local parameterization of the EEWL strategy space. 
For a dense-matrix implementation, joint-space payoff and gradient evaluation becomes the dominant computational cost as the problem size grows. In particular, the most expensive operations are the repeated construction of joint operators, matrix multiplications, and full/partial trace computations.
As a result, the dominant per-sweep cost scales as $O(N\Pi_{i\in\mathcal N}m_id_i^3)$.
When calculating the BR, as  $T \to 0$ and under the conditions of Theorem~\ref{thm:bundle-descent-usmea}, the updates become monotonic.
In this regime, convergence is guaranteed.
In contrast, when computing a fixed point $z^\star$, Theorem~\ref{theorem:convergece_all_local} provides a local convergence guarantee under spectral conditions on the differential $D\mathcal T(z^\star)$ at that fixed point. These assumptions are properties of the fixed point itself, not of the initialization. Accordingly, the theorem ensures local convergence only for iterates started sufficiently close to a fixed point satisfying these conditions.

\textbf{Implementation on quantum hardware and simulators:}
So far, we have used a general form of the EEWL quantum game with arbitrary local dimensions $d_i$ and applied tools from quantum mechanics to extend classical game theory in a fully mathematical way.  
For the implementation of the quantum games on quantum hardware and simulators, we specialize to 
$
    d_i = 2^{n_i},
$
where $n_i$ is the number of qubits controlled by player $i$, and the local Hilbert space is
$
    \mathcal{H}^{(i)} \cong (\mathbb{C}^2)^{\otimes n_i}.
$
Each local unitary $\mathbf{U}^{(i)}_{j_i}$ can be implemented as a parameterized quantum circuit built from qubit gates with parameter vectors $\boldsymbol{\theta}^{(i)}_{j_i}$.  
In this hybrid setting, the optimization algorithm still runs on a classical computer, while all required quantum probabilities are estimated on a quantum device.  
Since changing $\boldsymbol{\theta}^{(i)}_{j_i}$ preserves unitarity by construction, USMEA updates only the parameters $\boldsymbol{\theta}^{(i)}_{j_i}$ in the first step.  
In the second step, it updates the classical probabilities using per-action payoffs estimated from measurement statistics returned by the quantum hardware.  
The hardware cost in this setting is governed by three factors. (1) the total number of qubits $\sum_i n_i$, (2) the circuit depth required to approximate each $\mathbf{U}^{(i)}_{j_i}$ from its parameters $\boldsymbol{\theta}^{(i)}_{j_i}$, and (3) the number of measurement shots needed to estimate the payoffs with sufficient accuracy.  
This hybrid classical--quantum computing scheme shows that the proposed
framework can, in principle, be scaled from purely classical simulations
of small games to experimental implementations on near-term quantum
devices, where noise and hardware constraints may affect the learned
strategies.

\textbf{Reproducibility Statement:} All experiments in this paper are designed to be fully reproducible. Comprehensive details of the experimental setup, including implementation steps, algorithmic procedures, and parameter choices, are provided in the main text and in the supplementary material. ChatGPT and QuillBot were used to polish the language of this paper.

\section{Conclusion}
\label{sec:conclusion}
In this paper, we presented the EEWL framework for multiplayer quantum games with mixed quantum strategies.
We proposed USMEA as a learning algorithm for both unitary action sets and classical mixing probabilities, using an entropy-based loss function.
We believe that USMEA and the EEWL formulation provide a structured basis for further algorithmic development and stronger theoretical frameworks in quantum game dynamics.
We studied the existence of solutions in the EEWL setting and analyzed the convergence of USMEA under standard conditions for non-convex manifolds.
In the experimental results, we examined the convergence of our algorithm in four different quantum models.
While USMEA is effective as a learning procedure, it does not enumerate all Nash equilibria.
For future work, we plan to replace the softmax function with alternative models for probability assignment.
An important direction is to hybridize USMEA with methods such as homotopy or path-following to refine candidate fixed points.
One could also employ adaptive or stochastic Riemannian gradient methods to train the actions.
Beyond Nash equilibria, one can study correlated equilibria and regret in the EEWL setting.
Finally, it remains important to understand how noise and hardware limitations in quantum devices affect the learned strategies.

\bibliographystyle{IEEEtran}
\bibliography{references}

@IEEEtranBSTCTL{IEEEexample:BSTcontrol,
  CTLuse_article_number = "no",
  CTLuse_alt_spacing    = "no",
  CTLalt_stretch_factor = "4",
  CTLdash_repeated_names= "no",
  CTLname_format_string = "{f.~}{vv~}{ll}{, jj}",
  CTLname_latex_cmd     = "",
  CTLname_url_prefix    = "[Online]. Available:"
}

@article{eisert1999quantum,
	title={Quantum games and quantum strategies},
	author={Eisert, Jens and Wilkens, Martin and Lewenstein, Maciej},
	journal={Physical Review Letters},
	volume={83},
	number={15},
	pages={3077},
	year={1999},
	publisher={APS}
}

@article{eisert2000quantum,
	title={Quantum games},
	author={Eisert, Jens and Wilkens, Martin},
	journal={Journal of Modern Optics},
	volume={47},
	number={14-15},
	pages={2543--2556},
	year={2000},
	publisher={Taylor \& Francis}
}

@article{haykin2002adaptive,
  title={Adaptive filter theory},
  author={Haykin, Simon},
  journal={Prentice Hall google schola},
  volume={2},
  pages={333--346},
  year={2002}
}

@article{petersen2008matrix,
  title={The matrix cookbook},
  author={Petersen, Kaare Brandt and Pedersen, Michael Syskind and others},
  journal={Technical University of Denmark},
  volume={7},
  number={15},
  pages={510},
  year={2008}
}

@inproceedings{abrudan2005optimization,
  title={Optimization under unitary matrix constraint using approximate matrix exponential},
  author={Abrudan, Traian and Eriksson, Jan and Koivunen, Visa},
  booktitle={Conference Record of the Thirty-Ninth Asilomar Conference onSignals, Systems and Computers, 2005.},
  pages={242--246},
  year={2005},
  organization={IEEE}
}

@article{wiersema2023optimizing,
  title={Optimizing quantum circuits with Riemannian gradient flow},
  author={Wiersema, Roeland and Killoran, Nathan},
  journal={Physical Review A},
  volume={107},
  number={6},
  pages={062421},
  year={2023},
  publisher={APS}
}

@article{flitney2002introduction,
  title={An introduction to quantum game theory},
  author={Flitney, Adrian P and Abbott, Derek},
  journal={Fluctuation and Noise Letters},
  volume={2},
  number={04},
  pages={R175--R187},
  year={2002},
  publisher={World Scientific}
}

@article{edwards1954theory,
  title={The theory of decision making.},
  author={Edwards, Ward},
  journal={Psychological bulletin},
  volume={51},
  number={4},
  pages={380},
  year={1954},
  publisher={American Psychological Association}
}

@article{camerer2004advances,
  title={Advances in behavioral economics},
  author={Camerer, Colin F},
  journal={Russel Sage Foundation},
  year={2004}
}

@incollection{von2007theory,
  title={Theory of games and economic behavior: 60th anniversary commemorative edition},
  author={Von Neumann, John and Morgenstern, Oskar},
  booktitle={Theory of games and economic behavior},
  year={2007},
  publisher={Princeton university press}
}

@book{fudenberg1991game,
  title={Game theory},
  author={Fudenberg, Drew and Tirole, Jean},
  year={1991},
  publisher={MIT press}
}

@article{perez2024game,
	title={The Game Theory in Quantum Computers: A Review.},
	author={Ant{\'o}n, Raquel P{\'e}rez and S{\'a}nchez, Jos{\'e} Ignacio L{\'o}pez and Bellot, Alberto Corbi},
	journal={International Journal of Interactive Multimedia and Artificial Intelligence},
	volume={8},
	number={6},
	pages={6--14},
	year={2024}
}

@article{bostanci2022quantum,
  title={Quantum game theory and the complexity of approximating quantum Nash equilibria},
  author={Bostanci, John and Watrous, John},
  journal={Quantum},
  volume={6},
  pages={882},
  year={2022},
  publisher={Verein zur F{\"o}rderung des Open Access Publizierens in den Quantenwissenschaften}
}

@article{preskill2018quantum,
  title={Quantum computing in the NISQ era and beyond},
  author={Preskill, John},
  journal={Quantum},
  volume={2},
  pages={79},
  year={2018},
  publisher={Verein zur F{\"o}rderung des Open Access Publizierens in den Quantenwissenschaften}
}

@article{benjamin2001multiplayer,
  title={Multiplayer quantum games},
  author={Benjamin, Simon C and Hayden, Patrick M},
  journal={Physical Review A},
  volume={64},
  number={3},
  pages={030301},
  year={2001},
  publisher={APS}
}

@article{manton2002optimization,
  title={Optimization algorithms exploiting unitary constraints},
  author={Manton, Jonathan H},
  journal={IEEE transactions on signal processing},
  volume={50},
  number={3},
  pages={635--650},
  year={2002},
  publisher={IEEE}
}

@book{poundstone2011prisoner,
  title={Prisoner's dilemma},
  author={Poundstone, William},
  year={2011},
  publisher={Anchor}
}

@incollection{rapoport2018prisoner,
  title={Prisoner’s dilemma},
  author={Rapoport, Anatol},
  booktitle={The New Palgrave Dictionary of Economics},
  pages={10749--10753},
  year={2018},
  publisher={Springer}
}

@article{tucker1950two,
  title={A two-person dilemma},
  author={Tucker, Albert W},
  journal={Prisoner's Dilemma},
  year={1950},
  publisher={Doubleday}
}

@book{nielsen2010quantum,
  title={Quantum computation and quantum information},
  author={Nielsen, Michael A and Chuang, Isaac L},
  year={2010},
  publisher={Cambridge university press}
}

@article{flitney2002quantum,
  title={Quantum version of the Monty Hall problem},
  author={Flitney, Adrian P and Abbott, Derek},
  journal={Physical Review A},
  volume={65},
  number={6},
  pages={062318},
  year={2002},
  publisher={APS}
}

@inproceedings{xu2022experimental,
  title={Experimental implementation of quantum prisoner dilemma on IBM quantum computers},
  author={Xu, Kanyangzi and Wu, Zhe},
  booktitle={2022 15th International Conference on Advanced Computer Theory and Engineering (ICACTE)},
  pages={13--18},
  year={2022},
  organization={IEEE}
}

@article{flitney2004quantum,
  title={Quantum games with decoherence},
  author={Flitney, Adrian P and Abbott, Derek},
  journal={Journal of Physics A: Mathematical and General},
  volume={38},
  number={2},
  pages={449},
  year={2004},
  publisher={IOP Publishing}
}

@article{holtfort2024quantum,
  title={Quantum Economics: A Systematic Literature Review},
  author={Holtfort, Thomas and Horsch, Andreas and others},
  journal={SocioEconomic Challenges},
  volume={8},
  number={1},
  pages={62--77},
  year={2024},
  publisher={Academic Research and Publishing UG}
}

@article{li2014entanglement,
  title={Entanglement guarantees emergence of cooperation in quantum prisoner's dilemma games on networks},
  author={Li, Angsheng and Yong, Xi},
  journal={Scientific reports},
  volume={4},
  number={1},
  pages={6286},
  year={2014},
  publisher={Nature Publishing Group UK London}
}

@article{du2002playing,
  title={Playing prisoner's dilemma with quantum rules},
  author={Du, Jiangfeng and Xu, Xiaodong and Li, Hui and Zhou, Xianyi and Han, Rongdian},
  journal={Fluctuation and Noise Letters},
  volume={2},
  number={04},
  pages={R189--R203},
  year={2002},
  publisher={World Scientific}
}

@article{ullah2021game,
  title={Game Theory and Stock Investment},
  author={Ullah, Mahboob and Lakhan, Ghulam Rasool and Channa, Amanullah and Gul, Shabnam},
  journal={Multicultural Education},
  volume={7},
  number={6},
  pages={40--44},
  year={2021}
}

@article{ikeda2022theory,
  title={Theory of quantum games and quantum economic behavior},
  author={Ikeda, Kazuki and Aoki, Shoto},
  journal={Quantum Information Processing},
  volume={21},
  number={1},
  pages={27},
  year={2022},
  publisher={Springer}
}

@article{khan2021quantum,
  title={Quantum Prisoner’s Dilemma and high frequency trading on the quantum cloud},
  author={Khan, Faisal Shah and Bao, Ning},
  journal={Frontiers in Artificial Intelligence},
  volume={4},
  pages={769392},
  year={2021},
  publisher={Frontiers Media SA}
}

@article{hanauske2010doves,
  title={Doves and hawks in economics revisited: An evolutionary quantum game theory based analysis of financial crises},
  author={Hanauske, Matthias and Kunz, Jennifer and Bernius, Steffen and K{\"o}nig, Wolfgang},
  journal={Physica A: Statistical Mechanics and its Applications},
  volume={389},
  number={21},
  pages={5084--5102},
  year={2010},
  publisher={Elsevier}
}

@article{khrennikov2020quantum,
  title={Quantum-like modeling: cognition, decision making, and rationality},
  author={Khrennikov, Andrei},
  journal={Mind \& Society},
  volume={19},
  number={2},
  pages={307--310},
  year={2020},
  publisher={Springer}
}

@article{schmid2010experimental,
  title={Experimental implementation of a four-player quantum game},
  author={Schmid, Christian and Flitney, Adrian P and Wieczorek, Witlef and Kiesel, Nikolai and Weinfurter, Harald and Hollenberg, Lloyd CL},
  journal={New Journal of Physics},
  volume={12},
  number={6},
  pages={063031},
  year={2010},
  publisher={IOP Publishing}
}

@article{malvetti2024randomized,
  title={Randomized gradient descents on Riemannian manifolds: Almost sure convergence to global minima in and beyond quantum optimization},
  author={Malvetti, Emanuel and Arenz, Christian and Dirr, Gunther and Schulte-Herbr{\"u}ggen, Thomas},
  journal={arXiv preprint arXiv:2405.12039},
  year={2024}
}

@article{luchnikov2021riemannian,
  title={Riemannian geometry and automatic differentiation for optimization problems of quantum physics and quantum technologies},
  author={Luchnikov, Ilia A and Krechetov, Mikhail E and Filippov, Sergey N},
  journal={New Journal of Physics},
  volume={23},
  number={7},
  pages={073006},
  year={2021},
  publisher={IOP Publishing}
}

@book{absil2008optimization,
  title={Optimization algorithms on matrix manifolds},
  author={Absil, P-A and Mahony, Robert and Sepulchre, Rodolphe},
  year={2008},
  publisher={Princeton University Press}
}

@article{neff2014logarithmic,
  title={A logarithmic minimization property of the unitary polar factor in the spectral and Frobenius norms},
  author={Neff, Patrizio and Nakatsukasa, Yuji and Fischle, Andreas},
  journal={SIAM Journal on Matrix Analysis and Applications},
  volume={35},
  number={3},
  pages={1132--1154},
  year={2014},
  publisher={SIAM}
}

@book{watrous2018theory,
  title={The theory of quantum information},
  author={Watrous, John},
  year={2018},
  publisher={Cambridge university press}
}

@book{bhatia2013matrix,
  title={Matrix analysis},
  author={Bhatia, Rajendra},
  volume={169},
  year={2013},
  publisher={Springer Science \& Business Media}
}

@article{gao2017properties,
  title={On the properties of the softmax function with application in game theory and reinforcement learning},
  author={Gao, Bolin and Pavel, Lacra},
  journal={arXiv preprint arXiv:1704.00805},
  year={2017}
}

@article{rubinstein1997optimization,
  title={Optimization of computer simulation models with rare events},
  author={Rubinstein, Reuven Y},
  journal={European Journal of Operational Research},
  volume={99},
  number={1},
  pages={89--112},
  year={1997},
  publisher={Elsevier}
}

@article{attouch2013convergence,
  title={Convergence of descent methods for semi-algebraic and tame problems: proximal algorithms, forward--backward splitting, and regularized Gauss--Seidel methods},
  author={Attouch, Hedy and Bolte, J{\'e}r{\^o}me and Svaiter, Benar Fux},
  journal={Mathematical programming},
  volume={137},
  number={1},
  pages={91--129},
  year={2013},
  publisher={Springer}
}

@article{glicksberg1952further,
  title={A further generalization of the Kakutani fixed point theorem, with application to Nash equilibrium points},
  author={Glicksberg, Irving L},
  journal={Proceedings of the American Mathematical Society},
  volume={3},
  number={1},
  pages={170--174},
  year={1952},
  publisher={JSTOR}
}

@article{nash1950equilibrium,
  title={Equilibrium points in n-person games},
  author={Nash Jr, John F},
  journal={Proceedings of the national academy of sciences},
  volume={36},
  number={1},
  pages={48--49},
  year={1950},
  publisher={national academy of sciences}
}

@book{horn2012matrix,
  title={Matrix analysis},
  author={Horn, Roger A and Johnson, Charles R},
  year={2012},
  publisher={Cambridge university press}
}

@article{tseng2001convergence,
  title={Convergence of a block coordinate descent method for nondifferentiable minimization},
  author={Tseng, Paul},
  journal={Journal of optimization theory and applications},
  volume={109},
  number={3},
  pages={475--494},
  year={2001},
  publisher={Springer}
}

@article{grippo2000convergence,
  title={On the convergence of the block nonlinear Gauss--Seidel method under convex constraints},
  author={Grippo, Luigi and Sciandrone, Marco},
  journal={Operations research letters},
  volume={26},
  number={3},
  pages={127--136},
  year={2000},
  publisher={Elsevier}
}

@article{mckelvey1995quantal,
  title={Quantal response equilibria for normal form games},
  author={McKelvey, Richard D and Palfrey, Thomas R},
  journal={Games and economic behavior},
  volume={10},
  number={1},
  pages={6--38},
  year={1995},
  publisher={Elsevier}
}

@article{balter2024decision,
  title={Decision under ambiguity, composed optimization, and quantal response equilibria},
  author={Balter, Anne and Schumacher, Johannes M and Schweizer, Nikolaus},
  journal={Nikolaus, Decision under ambiguity, composed optimization, and quantal response equilibria (August 03, 2024)},
  year={2024}
}

@article{song2025fast,
  title={Fast state stabilization using deep reinforcement learning for measurement-based quantum feedback control},
  author={Song, Chunxiang and Liu, Yanan and Dong, Daoyi and Yonezawa, Hidehiro},
  journal={IEEE Transactions on Quantum Engineering},
  year={2025},
  publisher={IEEE}
}

@article{le2022dqra,
  title={DQRA: Deep quantum routing agent for entanglement routing in quantum networks},
  author={Le, Linh and Nguyen, Tu N},
  journal={IEEE Transactions on Quantum Engineering},
  volume={3},
  pages={1--12},
  year={2022},
  publisher={IEEE}
}

@article{abane2025entanglement,
  title={Entanglement routing in quantum networks: A comprehensive survey},
  author={Abane, Amar and Cubeddu, Michael and Mai, Van Sy and Battou, Abdella},
  journal={IEEE Transactions on Quantum Engineering},
  year={2025},
  publisher={IEEE}
}

@article{absil2006stable,
	title={On the stable equilibrium points of gradient systems},
	author={Absil, P.-A. and Kurdyka, K.},
	journal={Systems \& Control Letters},
	volume={55},
	number={7},
	pages={573--577},
	year={2006}
}

@inproceedings{lee2016gd,
	title={Gradient Descent Converges to Minimizers},
	author={Lee, Jason D. and Simchowitz, Max and Jordan, Michael I. and Recht, Benjamin},
	booktitle={Conference on Learning Theory},
	pages={1246--1257},
	year={2016}
}

@article{bugu2025entanglement,
	title={Entanglement as a Strategic Resource in Adversarial Quantum Games},
	author={Bugu, Sinan},
	journal={arXiv preprint arXiv:2510.22444},
	year={2025}
}

@article{du2003phase, 
	title={Phase-transition-like behaviour of quantum games},
	author={Du, Jiangfeng and Li, Hui and Xu, Xiaodong and Zhou, Xianyi and Han, Rongdian},
	journal={Journal of Physics A: Mathematical and General},
	volume={36},
	number={23},
	pages={6551--6562},
	year={2003}
}

@book{Rudin1991,
	author    = {Walter Rudin},
	title     = {Functional Analysis},
	edition   = {2},
	publisher = {McGraw-Hill},
	year      = {1991},
	series    = {International Series in Pure and Applied Mathematics}
}

@book{HornJohnson2013,
	author    = {Roger A. Horn and Charles R. Johnson},
	title     = {Matrix Analysis},
	edition   = {2},
	publisher = {Cambridge University Press},
	year      = {2013}
}

@book{Eldering2018,
	author    = {Jaap Eldering},
	title     = {Normally Hyperbolic Invariant Manifolds: The Noncompact Case},
	publisher = {Atlantis Press},
	series    = {Atlantis Studies in Dynamical Systems},
	volume    = {2},
	year      = {2018},
	doi       = {10.2991/978-94-6239-243-8},
	url       = {https://www.jaapeldering.nl/files/NHIM-noncompact-book.pdf}
}

@techreport{NippStoffer1992,
	author      = {K. Nipp and D. Stoffer},
	title       = {Attractive Invariant Manifolds for Maps: Existence, Smoothness and Continuous Dependence on the Map},
	institution = {Seminar f{\"u}r Angewandte Mathematik, ETH Z{\"u}rich},
	number      = {Research Report No. 92-11},
	year        = {1992},
	url         = {https://www.sam.math.ethz.ch/sam_reports/reports_final/reports1992/1992-11.pdf}
}

\newpage
\onecolumn

\begin{center}
\Large \textbf{Supplementary Materials:\\ Riemannian Optimization for Multi-Player Quantum Games on Product Unitary Manifolds}
\end{center}

\setcounter{figure}{0}
\setcounter{page}{1}
\makeatletter
\setcounter{equation}{0}
\renewcommand{\theequation}{S\arabic{equation}}
\setcounter{section}{0}
\renewcommand{\thesection}{S-\Roman{section}}
\setcounter{theorem}{0}
\renewcommand{\thetheorem}{S-\arabic{theorem}}
\setcounter{lemma}{0}
\renewcommand{\thelemma}{S-\arabic{lemma}}
\setcounter{proposition}{0}
\renewcommand{\theproposition}{S-\arabic{proposition}}

\section{Basics of quantum mechanics for game theorists}
\label{sec:quantumbasics}
To describe the quantum state of a system, we use a density matrix $\rho : \mathcal{H} \to \mathcal{H}$. The density matrix allows for the computation of outcome probabilities in quantum measurements.
A valid density matrix is a Hermitian, positive semidefinite, and satisfies
\begin{align}
\label{eq:densitymatrixprops}
\rho^\dagger = \rho, \quad  \Tr(\rho) = 1, \quad \tfrac{1}{d} \leq  \Tr(\rho^2) \leq 1 ,
\end{align}
The value $ \Tr(\rho^2)$ provides insight about the purity of the state.
For a pure quantum state, ${ \Tr(\rho^2) = 1}$, while for a mixed quantum state, ${ \Tr(\rho^2) < 1}$. In particular, for the maximally mixed state, one has ${ \Tr(\rho^2) = \tfrac{1}{d}}$.
In a finite-dimensional Hilbert space, an operator $\mathbf{A} : \mathcal{H} \to \mathcal{H}$ is a linear map and can be represented as a $d \times d$ complex matrix. 
In quantum mechanics, operators represent either physical observables or transformations of the quantum state. 
Unitary operators ($\mathbf{U}\mathbf{U}^\dagger = \mathbf{I}$, where $\mathbf{I}$ denotes the identity operator) can be used for the evolution of quantum states as $\rho' = \mathbf{U}\rho\mathbf{U}^\dag$.
Hermitian operators ($\mathbf{O} = \mathbf{O}^\dag$) correspond to measurable quantities (observables) and are guaranteed to have real eigenvalues. 
The commutator of two operators is defined as $[\mathbf{A},\mathbf{B}] = \mathbf{A}\mathbf{B} - \mathbf{B}\mathbf{A}$.
The expectation value of an observable $\mathbf{O}$ with respect to a quantum state $\rho$ is computed as
\begin{align}
    \braket{\mathbf{O}} = \Tr{\left( \mathbf{O} \rho \right) }.
\end{align}

For a pure quantum state, we can also use the Dirac notation. 
In Dirac notation, a pure quantum state  $\psi$ is represented by a ket, $\ket{\psi}$, in a Hilbert space spanned by a set of orthonormal basis kets $\{\ket{j}\}$.
Any state can be written as $\ket{\psi} = \sum_j c_j \ket{j}$, where the complex coefficients $c_j$ are probability amplitudes.
The quantity $|c_j|^2$ gives the probability of finding the system in basis state $\ket{j}$.
Each ket $\ket{\psi}$ has a vector counterpart in the chosen basis, often 
written as a column vector of complex amplitudes $(c_1, c_2, \ldots)^{T}$, 
and its dual bra $\bra{\psi}=\ket{\psi}^\dag$ corresponds to the conjugate transpose row vector. 
The inner product between two states is written as $\braket{\psi'|\psi}$.
The outer product $\ket{j'}\bra{j}$ forms a basis for the 
space of linear operators.
For composite systems, the joint state space is constructed using the Kronecker product of the subsystems' bases.
For example, for two subsystems $A$ and $B$, the  $\ket{j_Aj_B} = \ket{j_A} \otimes \ket{j_B}$ describes the joint basis.
For a mixed quantum state, the density matrix can be written in terms of pure states as
\begin{align}
    \rho = \sum_j \alpha_j \ket{\psi_j} \bra{\psi_j},
\end{align}
where $\alpha_j$ represents the classical probability of the system being in the pure state $\ket{\psi_j}$.

\textbf{Separable and Entangled states:}
In quantum mechanics, when dealing with a composite system composed of two or more subsystems, the overall state can be either separable or entangled. Consider that we have a system made up of subsystems $A$ and $B$ with the density matrix $\rho^{(AB)}$. The system is separable if its state can be expressed as
\begin{align}
    \rho^{(AB)} = \sum_j p_j   \rho^{(A)}_j \otimes \rho^{(B)}_j,
\end{align}
where $p_j$ is classical probabilities (with $\sum_j p_j = 1$), and $\rho^{(A)}_j$ and $\rho^{(B)}_j$ are the density matrices of subsystems $A$ and $B$, respectively. The symbol $\otimes$ denotes the Kronecker product with the properties
\begin{align}
    (\mathbf{A}\otimes\mathbf{B})(\mathbf{C}\otimes\mathbf{D})
    =(\mathbf{AC})\otimes(\mathbf{BD}).
\end{align}

In a separable state, the total system can be viewed as a statistical mixture of independent states.
If no such decomposition exists, the state $\rho^{(AB)}$ is entangled. Entangled states show a non-classical correlation between subsystems. This quantum correlation cannot be explained by classical probability theory or any local hidden variable model.

\section{Proofs and Additional Results}
\label{append:proofs}
In this section, we provide some additional theorems and missing proofs from the main part of our paper.
\begin{theorem}[Wirtinger's Calculus for Real-Valued Functions~\cite{petersen2008matrix,haykin2002adaptive}]
\label{theorem:Wirtinger-gradient}
Let $z$ be a complex variable, and let $f: \mathbb{C} \to \mathbb{R}$ be a real-valued function. Then, the gradient of $f$ with respect to $z$ is given by
\begin{align}
    \label{eq:Wirti}
    \nabla_z f(z, \bar z) = 2 \frac{\partial f(z, \bar z)}{\partial \bar z}.
\end{align}
\end{theorem}
\noindent\textbf{Proof:}
For a detailed derivation, see Refs.~\cite{petersen2008matrix,haykin2002adaptive}. 
In Eq.~\eqref{eq:Wirti}, the factor of two arises from Wirtinger’s calculus for real-valued functions with complex-valued parameters. This method treats the complex variable $z$ and its conjugate $\bar{z}$ as independent variables. 

\begin{lemma}
\label{lem:softmax-lip-l2-cited}
The softmax function is $L$-Lipschitz with respect to the $\|\cdot\|_2$ norm, with Lipschitz constant $L = 1/T$ as
\begin{align}
\big\|\sigma_T(\boldsymbol{\ell})-\sigma_T(\boldsymbol{\ell}')\big\|_{2}  \le  \frac{1}{T} \|\boldsymbol{\ell}-\boldsymbol{\ell}'\|_{2}.
\end{align}
\end{lemma}
 \noindent\textbf{Proof:}
For a detailed proof, see~\cite{gao2017properties}.

\begin{lemma}
\label{lemma:partial-trace-bound}
For any operator $\mathbf{X}$ on $\mathcal{H}^{(i)}\otimes\mathcal{H}^{(-i)}$,
\begin{align}
 \| \Tr_{-i}(\mathbf{X})\|_{F} \le  \sqrt{ d_{-i}} \|\mathbf{X}\|_{F}.
\end{align}
\end{lemma}
 \noindent\textbf{Proof:}
The Hilbert–Schmidt inner product between two operators is given by $\langle \mathbf{A},\mathbf{B}\rangle_{\text{HS}}= \Tr(\mathbf{A}^\dagger \mathbf{B})$.
For any operator $\mathbf{Y}^{(i)}: \mathcal{H}^{(i)} \to \mathcal{H}^{(i)}$, 
we have
\begin{align}
\langle  \Tr_{-i}(\mathbf{X}), \mathbf{Y}^{(i)}\rangle_{\text{HS}}
=  \Tr_{i} \left(
( \Tr_{-i}(\mathbf{X}))^\dag (\mathbf{Y}^{(i)}  \otimes \mathbf{I}^{(-i)})
\right)=
  \langle \mathbf{X}, \mathbf{Y}^{(i)}\otimes \mathbf{I}^{(-i)}\rangle_{\text{HS}}.
\end{align}
By duality, we have
\begin{align}
\label{eq:tr-frob}
\| \Tr_{-i}(\mathbf{X})\|_{F}
&= \sup_{\|\mathbf{Y}^{(i)}\|_{F}=1} |\langle  \Tr_{-i}(\mathbf{X}),\mathbf{Y}^{(i)}\rangle_{\text{HS}}| 
\nonumber\\
&= \sup_{\|\mathbf{Y}^{(i)}\|_{F}=1} |\langle \mathbf{X}, \mathbf{Y}^{(i)}\otimes \mathbf{I}^{(-i)}\rangle_{\text{HS}}|
\nonumber\\
&\le \|\mathbf{X}\|_{F} \sup_{\|\mathbf{Y}^{(i)}\|_{F}=1}\|\mathbf{Y}^{(i)}\otimes \mathbf{I}^{(-i)}\|_{F}.
\end{align}
Since Frobenius norm is multiplicative on tensor products~\cite{bhatia2013matrix}, we have
\begin{align}
\label{eq:kron-multi}
\|\mathbf{Y}^{(i)}\otimes \mathbf{I}^{(-i)}\|_{F}  =  \|\mathbf{Y}^{(i)}\|_{F} \|\mathbf{I}^{(-i)}\|_{F}  =  \sqrt{ d_{-i}}.
\end{align}
Combining Eq.~\eqref{eq:tr-frob} and Eq.~\eqref{eq:kron-multi}, we obtain
\begin{align}
    \| \Tr_{-i}(\mathbf{X})\|_{F} \le  \sqrt{ d_{-i}} \|\mathbf{X}\|_{F}.
\end{align}

\begin{proposition}
\label{prop:softmax-kl}
Consider that all actions are fixed, and player $i$ updates its probabilities by minimizing its loss function $L^{(i)}$ at temperature $T > 0$.  
Let $\mathbf{p}'^{(i)}$ be the probability vector that minimizes $L^{(i)}$, given by the softmax function~\cite{rubinstein1997optimization}.  
Then, for any other probability vector $\mathbf{p}^{(i)}$ in the simplex, the following improvement bound holds
\begin{align}
\label{eq:softmax-improve-L1-L2}
 L ^{(i)}(\mathbf{p}'^{(i)}) -  L ^{(i)}(\mathbf{p}^{(i)}) 
 \le  -\frac{T}{2} \left\| \mathbf{p}'^{(i)} - \mathbf{p}^{(i)} \right\|_2^2.
\end{align}
\end{proposition}
\noindent\textbf{Proof:}
At a minimum, $ L ^{(i)}$ satisfies
\begin{align}
 L ^{(i)}(\mathbf{p}'^{(i)}) =- T \log \sum_{j_i} \exp\left( \ell^{(i)}_{j_i} / T \right).
\end{align}
For any $\mathbf{p}^{(i)} \in \Delta^{m_i - 1}$, we compute the loss function difference as
\begin{align}
 L ^{(i)}(\mathbf{p}'^{(i)}) -  L ^{(i)}(\mathbf{p}^{(i)})
&= - \log \sum_{j_i} \exp\left( \ell^{(i)}_{j_i} / T \right)
+ \sum_{j_i} p^{(i)}_{j_i}   \ell^{(i)}_{j_i}
- T \sum_{j_i} p^{(i)}_{j_i} \log p^{(i)}_{j_i}.
\end{align}
Using softmax equation, $\ell^{(i)}_{j_i}$ can be expressed as
\begin{align}
\ell^{(i)}_{j_i} = T \left( \log \sum_{j_i} \exp\left( \ell^{(i)}_{j_i} / T \right) - \log p'^{(i)}_{j_i} \right),
\end{align}
we substitute the above equation into the loss function difference and obtain
\begin{align}
 L ^{(i)}(\mathbf{p}'^{(i)}) -  L ^{(i)}(\mathbf{p}^{(i)})
=  - T  \mathrm{KL}\left( \mathbf{p}^{(i)}  \|  \mathbf{p}'^{(i)} \right),
\end{align}
where $\mathrm{KL}(\cdot \| \cdot)$ denotes the Kullback–Leibler (KL) divergence
\begin{align}
\mathrm{KL}\left( \mathbf{p}^{(i)}  \|  \mathbf{p}'^{(i)} \right)
= \sum_{j_i} p^{(i)}_{j_i} \log \left( \frac{p^{(i)}_{j_i}}{p'^{(i)}_{j_i}} \right).
\end{align}
Applying Pinsker’s inequality and using $\|\cdot\|_1 \ge \|\cdot\|_2$, we obtain
\begin{align}
\mathrm{KL}\left( \mathbf{p}^{(i)}  \|  \mathbf{p}'^{(i)} \right)
 \ge  \frac{1}{2} \left\| \mathbf{p}^{(i)} - \mathbf{p}'^{(i)} \right\|_1^2
 \ge  \frac{1}{2} \left\| \mathbf{p}^{(i)} - \mathbf{p}'^{(i)} \right\|_2^2,
\end{align}
we conclude
\begin{align}
 L ^{(i)}(\mathbf{p}'^{(i)}) -  L ^{(i)}(\mathbf{p}^{(i)})
 \le  -\frac{T}{2} \left\| \mathbf{p}'^{(i)} - \mathbf{p}^{(i)} \right\|_2^2.
\end{align}

\noindent\textbf{Proof of Remark~\ref{remark:phaseandbounded}:}
\\
(1) Set $\mathbf{U}'^{(i)} = e^{i\theta_i}\mathbf{U}^{(i)}$. Then
\begin{align}
\mathbf{U}'^{(s)}\rho_0 \mathbf{U}'^{(s)\dag}
= e^{i\sum_i \theta_i} \mathbf{U}^{(s)} \rho_0 \mathbf{U}^{(s)\dag} e^{-i \sum_i \theta_i}
= \mathbf{U}^{(s)} \rho_0 \mathbf{U}^{(s)\dag},
\end{align}
so the payoff is unchanged.
\\
(2) Let $\rho = \mathbf{U}^{(s)}\rho_0 \mathbf{U}^{(s)\dag}$. 
$\rho$  is positive definite and  satisfies   $ \Tr(\rho)=1$. The map $\rho \mapsto  \Tr(\mathbf{R}_i \rho)$ is linear, and the set of density matrices  forms the convex hull of rank-1 projectors. Because  $ \Tr(\mathbf{R}_i \rho)$ is a convex combination of eigenvalues of $\mathbf{R}_i$ in an appropriate basis, we have
\begin{align}
\lambda_{\min}(\mathbf{R}_i)  \le   \Tr(\mathbf{R}_i \rho)  \le  \lambda_{\max}(\mathbf{R}_i).
\end{align}
For the operator-norm bound, we use the Schatten norm properties of operators and we obtain
\begin{align}
\big| \Tr(\mathbf{R}_i \rho)\big|
 \le  \|\mathbf{R}_i\|_{\text{op}} \|\rho\|
 =  \|\mathbf{R}_i\|_{\text{op}}  \Tr(\rho)
 =  \|\mathbf{R}_i\|_{\text{op}}.
\end{align}
Both bounds are tight. Equality happens when choosing $\rho$ as the projector onto an eigenvector of $\mathbf{R}_i$ with eigenvalue $\lambda_{\max}$ or $\lambda_{\min}$. For mixed strategies, the averaged state
\begin{align}
\bar\rho  =  \sum_{\mathbf{j}} p_{\mathbf{j}}  \mathbf{U}_{\mathbf{j}}^{(s)} \rho_0 \mathbf{U}_{\mathbf{j}}^{(s)\dag}
\end{align}
is again a density matrix (convex combination of density matrices). Replacing $\rho$ by $\bar\rho$ yields the same bounds for $\bar r^{(i)}$.

\noindent\textbf{Proof of Theorem~\ref{thm:EEWL-Glicksberg}:}\\
The unitary group $\mathcal{U}(d_i)$ is compact (closed and bounded in finite dimension).  
The set of mixed strategies for player $i$ is convex, closed, and compact.  
The expected payoff function is continuous, bounded, and linear in each player’s own mixed strategy.  
These properties satisfy the hypotheses of Glicksberg’s generalization of Nash’s theorem  
(compact action sets, continuous payoffs, and linearity in one’s own mixture)~\cite{glicksberg1952further}.  
Therefore, there is  at least one  mixed strategy Nash equilibrium in the EEWL quantum game.

\noindent\textbf{Proof of Theorem~\ref{theorem:calculateG}}:\\
(1)  In the first update of USMEA, $\mathbf{p}^{(i)}$ is fixed, so $H(\mathbf{p}^{(i)})$ remains constant and its derivative with respect to $\mathbf{U}^{(i)}_{j_i}$ vanishes.
Therefore, when computing the gradient of $ L ^{(i)}$, only the term $-\bar r^{(i)}$ contributes.  
Since Riemannian gradients are linear, we obtain (i).
\\
(2) The payoff function is real-valued but depends on complex-valued parameters.
First, we calculate the Euclidean gradient of $r^{(i)}(\mathbf{U}^{(s)})$ with respect to $\mathbf{U}^{(i)}$. 
We use the Wirtinger derivatives (Theorem~\ref{theorem:Wirtinger-gradient}) for real-valued functions with complex variables. We compute the gradient with respect to $\bar{U}_{mn}^{(i)}$. We treat $U_{mn}^{(i)}$ and $\bar{U}_{mn}^{(i)}$ as independent variables. 
For convenience, we define $\mathbf{X} = \mathbf{R}_i \mathbf{U} \rho_0$. The derivative of the real-valued function $r^{(i)}$ with respect to $\bar{U}_{mn}^{(i)}$ is then given by
\begin{align}
\frac{\partial  r^{(i)} (\mathbf{U}^{(\text{s})})}{\partial \bar U_{m,n}^{(i)}} &= 
\frac{\partial }{\partial \bar U_{m,n}^{(i)}} 
\Big(
\sum\limits_{\scriptstyle j_1,j_2,\ldots,j_N \hfill\atop
\scriptstyle j'_1,j'_2,\ldots,j'_N \hfill}
X_{j_1,\ldots,j_i,\ldots j_N,j'_1,\ldots,j'_i,\ldots,j'_N}
\bar U_{j_1,j'_1}^{(1)} \bar U_{j_2,j'_2}^{(2)}  \ldots  \bar U_{j_N,j'_N}^{(N)}
\Big) \nonumber \\
&= 
\sum\limits_{\scriptstyle j_1,j_2,\ldots,j_N \hfill\atop
\scriptstyle j'_1,j'_2,\ldots,j'_N \hfill}
X_{j_1,\ldots,j_i,\ldots j_N,j'_1,\ldots,j'_i,\ldots,j'_N}
\frac{\partial }{\partial \bar U_{m,n}^{(i)}} 
\Big(\bar U_{j_1,j'_1}^{(1)} \bar U_{j_2,j'_2}^{(2)}  \ldots   \bar U_{j_N,j'_N}^{(N)}
\Big). 
\end{align}
By applying the chain rule, we have
\begin{align}
\frac{\partial  r^{(i)} (\mathbf{U}^{(\text{s})})}{\partial \bar U_{m,n}^{(i)}} &= 
\sum\limits_{\scriptstyle j_1,j_2,\ldots,j_N \hfill\atop
\scriptstyle j'_1,j'_2,\ldots,j'_N \hfill}
R_{j_1,\ldots,j_i,\ldots j_N,j'_1,\ldots,j'_i,\ldots,j'_N}
\sum_{k=1}^N 
\Big( \frac{\partial }{\partial \bar U_{m,n}^{(i)}}  { \bar U_{j_k,j'_k}^{(k)}} \Big)  \prod_{l \ne k} {\bar  U_{j_l,j'_l}^{(l)}} 
\nonumber \\
&=
\sum\limits_{\scriptstyle j_1,j_2,\ldots,j_N \hfill\atop
\scriptstyle j'_1,j'_2,\ldots,j'_N \hfill}
X_{j_1,\ldots,j_i,\ldots j_N,j'_1,\ldots,j'_i,\ldots,j'_N}
\sum_{k=1}^N 
\Big(\delta_{m,j_k}\delta_{n,j'_k}\delta_{i,k}\Big)  \prod_{l \ne k} { \bar U_{j_l,j'_l}^{(l)}} 
\nonumber \\
&=
\sum\limits_{\scriptstyle\{ j_1,\ldots,j_{m-1},j_{m+1},\ldots,j_N\} \hfill\atop
\scriptstyle\{ j'_1,\ldots,j'_{n-1},j'_{n+1},\ldots,j'_N\} \hfill}
X_{j_1,\ldots,m,\ldots j_N,j'_1,\ldots,n,\ldots,j'_N}
\prod_{l \ne i} { \bar U_{j_l,j'_l}^{(l)}}. 
\end{align}
This expression simplifies to the $(m,n)$-th element of the partial trace and can be expressed as 
\begin{align}
      \frac{\partial  r^{(i)} (\mathbf{U}^{(\text{s})})}{\partial \bar U_{mn}^{(i)}} &= 
      \Big[  \Tr_{- i}\Big(\mathbf{X}  \mathbf{U}^{(\text{s})\dag}_{-i} \Big) \Big]_{mn}.
\end{align}
The derivative with respect to $\mathbf{U}^{(i)}$ is given by
\begin{align}
\frac{\partial  r^{(i)}(\mathbf{U}^{(\text{s})})}{\partial \bar{\mathbf{U}}^{(i)}} &= 
 \Tr_{-i}\Big(\mathbf{X} \mathbf{U}^{(\text{s})\dag}_{-i} \Big).
\end{align}
By substituting $\mathbf{X} = \mathbf{R}_i \mathbf{U} \rho_0$ into the equation and using Theorem~\ref{theorem:Wirtinger-gradient}, we have
\begin{align}
 \nabla_{\mathbf{U}^{(i)}}  {r}^{(i)} = 2 
\frac{\partial   r^{(i)} (\mathbf{U}^{(\text{s})})}{\partial \bar{\mathbf{U}}{(i)}} = 2
       \Tr_{- i}\Big(
      \mathbf{R}_i
        \mathbf{U}^{(\text{s})}
      \rho_0
      \mathbf{U}^{(\text{s})\dag}_{-i} \Big).
\end{align}
The Riemannian gradient of the payoff function with respect to the unitary operator $\mathbf{U}^{(i)}$ is given by (Ref.~\cite{abrudan2005optimization})
\begin{align}
    2 \mathbf{g}^{(i)} &=  
    \left( \nabla_{\mathbf{U}^{(i)}}  {r}^{(i)}\right)  \boldsymbol{U}^{(i)\dag} - 
    \boldsymbol{U}^{(i)} \left( \nabla_{\mathbf{U}^{(i)}}  {r}^{(i)}\right)^\dag,
\end{align}
which is an antisymmetric matrix. 
By performing straightforward  calculations, we obtain
\begin{align}
   \mathbf{g}^{(i)} =    \Tr_{-{i}}
    \big(
    \big[\mathbf{R}_i, 
     \mathbf{U}^{(\text{s})}
    \rho_0 
     \mathbf{U}^{{(\text{s})}\dag} \big]\big).
\end{align}
A similar calculation for the mixed strategy leads to
\begin{align}
\mathbf{G}^{(i)}_{j_i}&=
\sum_{\mathbf{j}_{-i}} p_{j_i;\mathbf{j}_{-i}}
 \Tr_{-i}\Big([\mathbf{R}_i , \mathbf{U}^{(s)}_{j_i;\mathbf{j}_{-i}} \rho_0 \mathbf{U}^{(s)\dag}_{j_i;\mathbf{j}_{-i}} ]\Big).
\end{align}

\noindent\textbf{Proof of Theorem~\ref{thm:block-lip}:} \\
Consider two unitary actions $\mathbf{U}^{(i)}_{j_i}$ and $\mathbf{U}'^{(i)}_{j_i}$ for player $i$ with action $j_i$. These actions can be connected via a geodesic path:
$
\mathbf{U}^{(i)}(t) = \exp(t \mathbf{X}^{(i)}) \mathbf{U}^{(i)},
$
where $\mathbf{X}^{(i)}= -\mathbf{X}^{(i)\dag}$  such that $\|\mathbf{X}^{(i)}\|_F = d_{\mathcal{U}}\big(\mathbf{U}^{(i)},\mathbf{U}'^{(i)}\big)$ and $t \in [0,1]$.
For each $t$, let $\mathbf{U}^{(s)}(t)$  the joint strategy profile and the evolved density matrix 
$
\rho(t) = \mathbf{U}^{(s)}(t)  \rho_0  \mathbf{U}^{(s)\dagger}(t).
$
The Riemannian gradient of the payoff over  $\mathbf{U}^{(i)}(t)$ (block $i$) at time $t$ is
\begin{align}
\mathbf{g}^{(i)}(t)
= \Tr_{-i} \left( \left[\mathbf{R}_i, \rho(t)\right] \right).
\end{align}
The Riemannian gradient of expected reward over action $j_i$ is then given by
\begin{align}
\mathbf{G}^{(i)}_{j_i}(t)
= \sum_{\mathbf{j}_{-i}} p_{j_i,\mathbf{j}_{-i}}
  \mathbf{g}^{(i)}_{j_i,\mathbf{j}_{-i}}(t).
\end{align}
\\
(1) By differentiating $\mathbf{g}^{(i)}(t)$ along the geodesic, we obtain:
\begin{align}
\label{eq:dif-g}
\frac{d}{dt} \mathbf{g}^{(i)}(t)
&= \Tr_{-i} \left( \left[
\mathbf{R}_i , \mathbf{X}^{(i)} \rho(t)
\right] + \left[
\mathbf{R}_i , \rho(t) \mathbf{X}^{(i)\dagger}
\right] \right) \nonumber\\
&= \Tr_{-i} \left( \left[
\mathbf{R}_i , \mathbf{X}^{(i)} \rho(t)
\right] - \left[
\mathbf{R}_i , \rho(t) \mathbf{X}^{(i)}
\right] \right) \nonumber\\
&= \Tr_{-i} \left( \left[\mathbf{R}_i, \left[\mathbf{X}^{(i)}, \rho(t)\right] \right] \right).
\end{align}
By Lemma~\ref{lemma:partial-trace-bound} and the bound
$\|[\mathbf{A}, \mathbf{B}]\|_F \le 2 \|\mathbf{A}\|_{\mathrm{op}} \|\mathbf{B}\|_F$, we get
\begin{align}
\left\| \frac{d}{dt} \mathbf{g}^{(i)}(t) \right\|_F
&\le \sqrt{ d_{-i}}  \left\| \left[ \mathbf{R}_i, \left[\mathbf{X}^{(i)},  \rho(t)\right] \right] \right\|_F \nonumber\\
&\le 2 \sqrt{ d_{-i}}  \|\mathbf{R}_i\|_{\mathrm{op}}  \left\| [\mathbf{X}^{(i)},  \rho(t)] \right\|_F \nonumber\\
&\le 4 \sqrt{ d_{-i}}  \|\mathbf{R}_i\|_{\mathrm{op}}  \|\mathbf{X}^{(i)}\|_{\mathrm{op}}  \|\rho(t)\|_F \nonumber\\
&\le 4 \sqrt{ d_{-i}}  \|\mathbf{R}_i\|_{\mathrm{op}}  \|\mathbf{X}^{(i)}\|_{\mathrm{op}}  \|\rho_0\|_F \nonumber\\
&\le 4 \sqrt{ d_{-i}}  \|\mathbf{R}_i\|_{\mathrm{op}}  \|\mathbf{X}^{(i)}\|_F  \|\rho_0\|_F.
\end{align}
Using Jensen's inequality for expectations, $\|\mathbb{E}(\mathbf{A})\|_F \le \mathbb{E}(\|\mathbf{A}\|_F)$, and integrating over $t \in [0, 1]$, we obtain
\begin{align}
\left\| \mathbf{G}^{(i)}_{j_i}(1) - \mathbf{G}^{(i)}_{j_i}(0) \right\|_F
&\le \int_0^1 \left\| \frac{d}{dt} \mathbf{G}^{(i)}_{j_i}(t) \right\|_F dt \nonumber\\
&\le 4  \sqrt{ d_{-i}}  \|\mathbf{R}_i\|_{\mathrm{op}}  \|\rho_0\|_F  \|\mathbf{X}^{(i)}\|_F.
\end{align}
 Finally, recalling the distance defined in Eq.~\eqref{eq:distanceunitary} as $d_{\mathcal{U}} = \|\mathbf{X}^{(i)}\|_F$, we conclude the proof of part (i) as
\begin{align}
\left\| \mathbf{G}^{(i)}_{j_i}(\mathbf{U}^{(i)}_{j_i}) - \mathbf{G}^{(i)}_{j_i}(\mathbf{U}'^{(i)}_{j_i}) \right\|_F
\le A_i  d_{\mathcal{U}}\left(\mathbf{U}^{(i)}_{j_i},  \mathbf{U}'^{(i)}_{j_i}\right),
\end{align}
where
\begin{align}
A_i = 4  \sqrt{ d_{-i}}  \|\mathbf{R}_i\|_{\mathrm{op}}  \|\rho_0\|_F.
\end{align}
\\
(2) The per-action payoff $\ell^{(i)}_{j_i}(t)$ can be expressed as
\begin{align}
\ell^{(i)}_{j_i}(t)
&= \sum_{\mathbf{j}_{-i}} p_{j_i, \mathbf{j}_{-i}}   r^{(i)}\left( \mathbf{U}^{(s)}_{j_i, \mathbf{j}_{-i}}(t) \right) \nonumber \\
&= \sum_{\mathbf{j}_{-i}} p_{j_i, \mathbf{j}_{-i}}   \Tr \left( \mathbf{R}_i   \rho_{j_i, \mathbf{j}_{-i}}(t) \right).
\end{align}
Differentiating $\ell^{(i)}_{j_i}(t)$ along the geodesic yields
\begin{align}
\frac{d}{dt}   \ell^{(i)}_{j_i}(t)
&= \sum_{\mathbf{j}_{-i}} p_{j_i, \mathbf{j}_{-i}}   \Tr \left( \mathbf{R}_i \left[ \mathbf{X}^{(i)}, \rho_{j_i, \mathbf{j}_{-i}}(t) \right] \right) \nonumber \\
&= - \sum_{\mathbf{j}_{-i}} p_{j_i, \mathbf{j}_{-i}}   \Tr \left( \left[ \mathbf{R}_i, \rho_{j_i, \mathbf{j}_{-i}}(t) \right]   \mathbf{X}^{(i)} \right) \nonumber \\
&= \sum_{\mathbf{j}_{-i}} p_{j_i, \mathbf{j}_{-i}}   \Tr \left( \mathbf{X}^{(i)\dagger}   [\mathbf{R}_i, \rho_{j_i, \mathbf{j}_{-i}}(t)] \right) \nonumber \\
&= \left\langle \mathbf{X}^{(i)}, \sum_{\mathbf{j}_{-i}} p_{j_i, \mathbf{j}_{-i}}   \Tr_{-i} \left( [\mathbf{R}_i, \rho_{j_i, \mathbf{j}_{-i}}(t)] \right) \right\rangle_{\text{HS}} \nonumber \\
&= \left\langle \mathbf{X}^{(i)},   \mathbf{G}^{(i)}_{j_i}(t) \right\rangle_{\text{HS}}.
\end{align}
\\
Applying the Cauchy–Schwarz inequality gives
\begin{align}
\label{eq:dldt}
\left| \frac{d}{dt}   \ell^{(i)}_{j_i}(t) \right|
= \left| \left\langle \mathbf{X}^{(i)},   \mathbf{G}^{(i)}_{j_i}(t) \right\rangle_{\text{HS}} \right|
\le \big\| \mathbf{G}^{(i)}_{j_i}(t) \big\|_F  \|\mathbf{X}^{(i)}\|_F.
\end{align}
\\
To bound $\|\mathbf{G}^{(i)}_{j_i}(t)\|_F$, we first bound the pure-strategy case. For each $\mathbf{j}_{-i}$, Lemma~\ref{lemma:partial-trace-bound} and commutator norms imply
\begin{align}
\left\| \Tr_{-i} \left( [\mathbf{R}_i, \rho_{j_i, \mathbf{j}_{-i}}(t)] \right) \right\|_F
&\le \sqrt{ d_{-i}}   \left\| [\mathbf{R}_i, \rho_{j_i, \mathbf{j}_{-i}}(t)] \right\|_F \nonumber \\
&\le 2  \sqrt{ d_{-i}}   \|\mathbf{R}_i\|_{\mathrm{op}}   \|\rho_{j_i, \mathbf{j}_{-i}}(t)\|_F \nonumber \\
&\le 2  \sqrt{ d_{-i}}   \|\mathbf{R}_i\|_{\mathrm{op}}   \|\rho_0\|_F.
\end{align}
\\
Taking expectation over $\mathbf{j}_{-i}$ and using Jensen’s inequality
\begin{align}
\label{eq:absG}
\left\| \mathbf{G}^{(i)}_{j_i}(t) \right\|_F
\le 2  \sqrt{ d_{-i}}  \|\mathbf{R}_i\|_{\mathrm{op}}   \|\rho_0\|_F.
\end{align}
\\
By integrating Eq.~\eqref{eq:dldt} from $t = 0$ to $t = 1$, and applying the bound from Eq.~\eqref{eq:absG}, we obtain
\begin{align}
\left| \ell^{(i)}_{j_i}(1) - \ell^{(i)}_{j_i}(0) \right|
&\le \int_0^1 \left\| \mathbf{G}^{(i)}_{j_i}(t) \right\|_F   \|\mathbf{X}^{(i)}\|_F   dt \nonumber \\
&\le 2  \sqrt{ d_{-i}}  \|\mathbf{R}_i\|_{\mathrm{op}}   \|\rho_0\|_F   \|\mathbf{X}^{(i)}\|_F.
\end{align}
\\
Finally, substituting $d_{\mathcal{U}} = \|\mathbf{X}^{(i)}\|_F$ and defining
\begin{align}
M_i = 2  \sqrt{ d_{-i}}  \|\mathbf{R}_i\|_{\mathrm{op}}   \|\rho_0\|_F,
\end{align}
we conclude Eq.~\eqref{eq:lLipschitz}.

\noindent\textbf{Proof of Theorem~\ref{thm:bundle-descent-usmea}}\\
First, fix $\mathbf{p}^{(i)}$ and modify $\{\mathbf{U}^{(i)}_{j_i}\}$.
Using ~\cite[Lemma~3.7]{malvetti2024randomized}, we have
\begin{align}
\label{eq:unitary-gain}
 L ^{(i)}\left((\{\mathbf{U}'^{(i)}_{j_i}\}, \mathbf{p}^{(i)});z^{(-i)}\right)
-  L ^{(i)}\left((\{\mathbf{U}^{(i)}_{j_i}\}, \mathbf{p}^{(i)});z^{(-i)}\right) 
\le 
- \eta_i \left(1-\tfrac{1}{2}\eta_i A_i\right)\Delta_{\mathcal{U}}^{(i)2},
\end{align}
Next, update $\mathbf{p}^{(i)}$ while keeping the $\{\mathbf{U}'^{(i)}_{j_i}\}$ fixed. Using Proposition~\ref{prop:softmax-kl}, we get
\begin{align}
\label{eq:prob-gain}
 L ^{(i)}\left((\{\mathbf{U}'^{(i)}_{j_i}\}, \mathbf{p}'^{(i)});z^{(-i)}\right)
-  L ^{(i)}\left((\{\mathbf{U}'^{(i)}_{j_i}\}, \mathbf{p}^{(i)});z^{(-i)}\right)
\le -\frac{T}{2} \left\| \mathbf{p}'^{(i)} - \mathbf{p}^{(i)} \right\|_2^2.
\end{align}
Adding inequalities Eqs.~\eqref{eq:unitary-gain} and~\eqref{eq:prob-gain} gives the bound in Eq.~\eqref{eq:bundle-descent-usmea}.
Since both terms on the right-hand side are nonnegative for $\eta_i \le 1/A_i$, monotonicity of $ L ^{(i)}$ follows.

\noindent\textbf{Proof of Theorem~\ref{theorem:BR-USMEA}:}\\
Fix $z^{(-i)}$.
The function $ L ^{(i)} \left(z^{(i)};z^{(-i)}\right)$ is bounded as
\begin{align}
   \left\|  L ^{(i)} \left(z^{(i)};z^{(-i)}\right) \right\|
    &\le \left\|\sum_{j_i} p^{(i)}_{j_i} \ell^{(i)}_{j_i} \right\|
        + \left\| T\sum_{j_i} p^{(i)}_{j_i}\log p^{(i)}_{j_i} \right\|
        \nonumber \\
    &\le \sum_{j_i} p^{(i)}_{j_i} \|R_i\|_{\op} 
        + \left\| T\sum_{j_i} p^{(i)}_{j_i}\log p^{(i)}_{j_i} \right\| \nonumber\\
    &\le \|R_i\|_{\op} + T\log m_i.
\end{align}
Moreover, by Theorem~\ref{thm:bundle-descent-usmea}, if $\eta_i \le 1/A_i$, then after each update at step $k$,
\begin{align}
\label{eq:bundle-descent-usmea-the}
 L ^{(i)} \left(z^{(i)[t+1]};z^{(-i)}\right)
-
 L ^{(i)} \left(z^{(i)[t]};z^{(-i)}\right)
&\le
-\eta_i \left(1 - \tfrac{1}{2} \eta_i A_i \right) 
    \sum_{j=1}^{m_i} \big\| \mathbf{G}^{(i)[t]}_{j_i} \big\|_F^2
-
\frac{T}{2} \big\| \mathbf{p}^{(i)[t+1]} - \mathbf{p}^{(i)[t]} \big\|_2^2 \\
&\le 
\eta_i \left(1 - \tfrac{1}{2} \eta_i A_i \right) 
    \sum_{j=1}^{m_i} \big\| \mathbf{G}^{(i)[t]}_{j_i} \big\|_F^2 .
\nonumber
\end{align}
Since the last expression is non-positive, 
$ L ^{(i)}(z^{(i)[t]};z^{(-i)})$ 
is non-decreasing.  
Summing Eq.~\eqref{eq:bundle-descent-usmea-the} over $k$ and using the boundedness of $ L ^{(i)}$, we obtain
\begin{align}
\sum_{k=0}^\infty \sum_{j_i=1}^{m_i}\big\|\mathbf{G}^{(i)[t]}_{j_i}\big\|_F^2 < \infty,
\quad \Rightarrow \quad 
\liminf_{k\to\infty} \max_{j_i} 
\big\|\mathbf{G}^{(i)[t]}_{j_i}\big\|_F = 0.
\end{align}
Since $ L ^{(i)}$ is real-analytic on the compact manifold $\mathcal{M}$, 
it satisfies the Kurdyka--Łojasiewicz (KL) property.  
By the KL convergence principle~\cite{absil2008optimization,attouch2013convergence},  
the sequence $\{z^{(i)[t]}\}_t$ converges to a critical point 
$z^{(i)*}$ of $ L ^{(i)}$, i.e.,
\begin{align}
\mathbf{G}^{(i)}_{j_i} \left(\mathbf{U}^{(i)*}_{j_i}, \mathbf{p}^{(i)*}\right)=0,
\quad \forall j_i.
\end{align}
Consequently, the optimal probabilities are given by the softmax.

\noindent\textbf{Proof of Corollary~\ref{cor:oscillation}:}\\
Since $H^{[t]} \in (0,\log m_i)$ for all $k$, for any $k$ and $l$ we have
\begin{align}
\big|\bar r^{(i)[t]} - \bar r^{(i)[l]}\big|
&= \big|\big( L ^{(i)[t]} -  L ^{(i)[l]}\big) + T\big(H^{[t]} - H^{[l]}\big)\big| \nonumber\\
&\le | L ^{(i)[t]} -  L ^{(i)[l]}| + T |H^{[t]} - H^{[l]}|.
\end{align}
If $\{ L ^{(i)[t]}\}$ converges, the first term vanishes.  
Thus,
\begin{align}
\limsup_{k} \bar r^{(i)[t]} - \liminf_{k} \bar r^{(i)[t]} 
 \le  T\big(\max H - \min H\big) 
 \le  T \log m_i.
\end{align}

\noindent\textbf{Proof of Theorem~\ref{theorem:exist-fp}}\\
Let $\{z^{[t,0]}\}_{t \ge 0}$ be the sequence of iterates generated by the  iteration map, initialized at $z^{[0,0]}$.
We first prove the existence of accumulation points (i.e., limit points).  
Since the manifold $\mathcal{M}$ is compact and the iteration sequence $\{z^{[t,0]}\} \subset \mathcal{M}$ is infinite,  
by sequential compactness there exists a convergent subsequence $\{z^{[ t,0]}\}_{t \ge 0}$ and a point $z^\star \in \mathcal{M}$ such that  
$
\lim_{t \to \infty} z^{[ t,0]} = z^\star,
$
as guaranteed by standard results in manifold optimization~\cite{absil2008optimization}.  
Therefore, $z^\star$ is an accumulation point of the full sequence $\{z^{[t,0]}\}$.
Next, we argue that the map $\mathcal{T} = \mathcal{T}_N \circ \cdots \circ \mathcal{T}_1$ is continuous.  
This follows directly from Theorem~\ref{thm:block-lip}, which ensures that the Riemannian gradient update over the unitary group is Lipschitz continuous,  
and the softmax update over the simplex is smooth. As a result, each block update map $\mathcal{T}_i$ is continuous, and so is their composition $\mathcal{T}$ is also continuous.
\\
Finally, we prove that any accumulation point must be a fixed point.  
Assume, for contradiction, that $z^\star$ is not a fixed point of $\mathcal{T}$, meaning that for some player $i$,  
the block update at $z^\star$ is not stationary.  
That is, either $\nabla_{U^{(i)}_{j_i}}  L ^{(i)}(z^\star) \ne 0$ or $p^{(i)} \ne \sigma_T(\ell^{(i)})$.  
From Theorem~\ref{thm:bundle-descent-usmea}, it then follows that there exists a neighborhood around such a point $\tilde z^\star_i$  
in which player $i$'s update yields a strict decrease in loss of at least some fixed constant $\delta > 0$.
\\
Because $\{z^{[ t,0]}\}$ converges to $z^\star$, for sufficiently large $t$,  
 $z^{[t+1,i-1]}$ lies within this neighborhood.  
As a result, the loss satisfies
\begin{align}
 L ^{(i)}(z^{[t+1, i]}) \le  L ^{(i)}(z^{[t+1, i-1]}) - \delta,
\qquad \text{for large } t.
\end{align}
This contradicts the convergence of the loss sequence $ L ^{(i)}(z^{[ t, i]})$ as established in Theorem~\ref{theorem:BR-USMEA}.  
Hence, each local update at the accumulation point must be block-stationary. That is,
\begin{align}
\mathcal{T}_i(\tilde z^\star_{i-1}) = \tilde z^\star_i, \qquad \text{with } \tilde z^\star_0 = z^\star.
\end{align}
Recursively applying this for all $i \in \mathcal N$ yields
$
\mathcal{T}(z^\star) 
= \mathcal{T}_N \circ \cdots \circ \mathcal{T}_1(z^\star) 
= \tilde z^\star_N 
= z^\star
$,
which shows that $z^\star$ is a fixed point of $\mathcal{T}$.

\noindent\textbf{Proof of Theorem~\ref{theorem:convergece_all_local}:}\\
The proof proceeds in three steps. First, we identify the neutral directions associated with the eigenvalue $1$ and separate them from the remaining spectrum. Second, we introduce local coordinates near the fixed-point submanifold and derive a suitable local representation of $\mathcal T$. Third, we show convergence of the iterates.
\\
\textbf{Step 1:}
Since $\mathcal M^{(\mathrm{FP})}$ is a manifold of fixed points, for any smooth curve $z(t)\subset \mathcal M^{(\mathrm{FP})}$ with $z(0)=z^\star$ we have
$
	\mathcal T(z(t))=z(t).
$
Differentiating at $t=0$ gives
$
	D\mathcal T(z^\star)z'(0)=z'(0).
$
Hence every tangent vector to $\mathcal M^{(\mathrm{FP})}$ at $z^\star$ is an eigenvector of $D\mathcal T(z^\star)$ associated with the eigenvalue $1$. Therefore,
$
	T_{z^\star}\mathcal M^{(\mathrm{FP})} \subseteq \mathcal E_1,
$
where $\mathcal E_1$ denotes the eigenspace of $D\mathcal T(z^\star)$ corresponding to $\lambda=1$.
By assumption~(1),
$
	\dim T_{z^\star}\mathcal M^{(\mathrm{FP})}=l.
$
By assumption~(2), the eigenspace $\mathcal E_1$ also has dimension $l$. Thus,
$
	T_{z^\star}\mathcal M^{(\mathrm{FP})}=\mathcal E_1.
$
Therefore, the neutral directions of the linearized USMEA dynamics are exactly the tangent directions of the fixed-point submanifold $\mathcal M^{(\mathrm{FP})}$.
Since the eigenvalue $1$ is semisimple and all other eigenvalues satisfy $|\lambda|<1$, the tangent space $T_{z^\star}\mathcal M$ admits the invariant decomposition
\begin{align}
	T_{z^\star}\mathcal M = \mathcal E_1 \oplus \mathcal E_s,
\end{align}
where $\mathcal E_s$ is the invariant subspace corresponding with the eigenvalues strictly inside the unit disk. Since $\mathcal E_1=T_{z^\star}\mathcal M^{(\mathrm{FP})}$, the subspace $\mathcal E_s$ can be regarded as the contracting normal space.
\\
\textbf{Step 2:} 
Choose local coordinates $(u,v)$ near $z^\star$ such that
$
T_{z^\star}\mathcal M=\mathcal E_1\oplus \mathcal E_s,
$
where
$
u\in \mathbb R^l
$
parametrizes directions tangent to $\mathcal M^{(\mathrm{FP})}$ and
$
v\in \mathbb R^r
$
parametrizes directions in $\mathcal E_s$. 
Since $\mathcal M^{(\mathrm{FP})}$ is a $C^1$ embedded submanifold, these coordinates may be chosen so that
$
\mathcal M^{(\mathrm{FP})}=\{(u,v): v=0\}.
$
In these coordinates, the local representative of $\mathcal T$ can be written as \cite{Eldering2018,NippStoffer1992}
\begin{align}
	\widetilde{\mathcal T}(u,v)=\bigl(u+a(u,v), Bv+b(u,v)\bigr),
\end{align}
with
$
	a(u,0)=0, b(u,0)=0
$,
and where the spectrum of $B$ consists of the eigenvalues of $D\mathcal T(z^\star)$ other than $1$.
By assumption~(3), every eigenvalue of $B$ lies strictly inside the unit disk. Hence the spectral radius satisfies
$r(B)<1$.
Since the space is finite-dimensional, there exists an equivalent norm $\|\cdot\|$ on $\mathbb R^r$ and a constant $q\in(0,1)$ such that \cite{HornJohnson2013,Rudin1991}
\begin{align}
	\|Bv\|\le q\|v\|,
	\qquad \forall v\in \mathbb R^r.
\end{align}
\\
Moreover, since $\widetilde{\mathcal T}$ is smooth, after shrinking the neighborhood, there exists a constant $c>0$ such that
\begin{align}
	\|a(u,v)\|\le c\|v\|,
	\qquad
	\|b(u,v)\|\le c\|v\|^2.
\end{align}
\\
\textbf{Step 3:} Convergence of the iterates:\\
Let
\begin{align}
	(u^{[k+1]},v^{[k+1]})=\widetilde{\mathcal T}(u^{[k]},v^{[k]}).
\end{align}
Then
\begin{align}
	v^{[k+1]}=Bv^{[k]}+b(u^{[k]},v^{[k]}),
\end{align}
and therefore
\begin{align}
	\|v^{[k+1]}\|
	\le
	q\|v^{[k]}\|+c\|v^{[k]}\|^2.
\end{align}
Choose $\theta\in(q,1)$. By shrinking the neighborhood once more, we may assume that
$
c\|v\|\le \theta-q
$
throughout the neighborhood. Hence,
\begin{align}
	\|v^{[k+1]}\|\le \theta \|v^{[k]}\|.
\end{align}
It follows that $v^{[k]}\to 0$ geometrically.
\\
For the tangential component, we have
\begin{align}
	u^{[k+1]}-u^{[k]}=a(u^{[k]},v^{[k]}),
\end{align}
so
\begin{align}
	\|u^{[k+1]}-u^{[k]}\|\le c\|v^{[k]}\|.
\end{align}
Since $\|v^{[k]}\|$ decays geometrically, the series
\begin{align}
	\sum_{k=0}^{\infty}\|u^{[k+1]}-u^{[k]}\|
\end{align}
converges. Therefore, $(u^{[k]})$ is a Cauchy sequence, and hence
$
	u^{[k]}\to u^{\infty}
$
for some $u^{\infty}$.
Combining this with $v^{[k]}\to 0$, we obtain
\begin{align}
	(u^{[k]},v^{[k]})\to (u^\infty,0)\in \mathcal M^{(\mathrm{FP})}.
\end{align}
Returning to the original manifold coordinates yields
$
z^{[k]}\to z^\infty\in \mathcal M^{(\mathrm{FP})}.
$
This proves the result.

\section{Derivation of the block differentials $D\mathcal T_i^{U}(z^\star)$ and $D\mathcal T_i^{p}(z^\star)$}
\label{supp:DT}
In this section, we compute the block differentials $D\mathcal T_i^{U}(z^\star)$ and $D\mathcal T_i^{p}(z^\star)$ by linearizing them at the fixed point $z^\star$. 
\subsection*{Construction of $D\mathcal T_i^{U}(z^\star)$}
Since $\mathcal T_i^{U}$ modifies only the unitary actions of player $i$, its differential is the identity on all blocks except the unitary-action block associated with player $i$.
\\
At the fixed point $z^\star$, the linearization takes the form
\begin{align}
	\delta z^{(i)+}_U
	=
	\delta z^{(i)}_U
	+
	\eta_i
	\left(
	\sum_{i'=1}^N \mathbf H_{ii'} \delta z^{(i')}_U
	+
	\sum_{i'=1}^N \mathbf B_{ii'} \delta \mathbf p^{(i')}
	\right),
\end{align}
where $\delta z^{(i)}_U$ collects the perturbations of all unitary actions of player $i$, $\mathbf H_{ii'}$ denotes the block assembled from the matrices $\mathbf H^{(i,i')}_{j_i,j_{i'}}$, and $\mathbf B_{ii'}$ denotes the block assembled from the vectors $\mathbf B^{(i,i')}_{j_i,j_{i'}}$.
\\
Therefore, $D\mathcal T_i^{U}(z^\star)$ is the identity on every block except the row corresponding to the unitary-action variables of player $i$. On that row, it is given by
\begin{align}
	\bigl[
	\eta_i \mathbf H_{i1}, 
	\eta_i \mathbf B_{i1}, 
	\dots, 
	\mathbf I+\eta_i \mathbf H_{ii}, 
	\eta_i \mathbf B_{ii}, 
	\dots, 
	\eta_i \mathbf H_{iN}, 
	\eta_i \mathbf B_{iN}
	\bigr].
\end{align}
In the following, we construct the blocks $\mathbf H_{ii'}$ and $\mathbf B_{ii'}$ in detail.
\\
\subsubsection{Construction of $\mathbf H_{ii'}$}
We now compute the derivative of the action gradient with respect to the action variable $\mathbf U^{(i')}_{j_{i'}}$. For a tangent perturbation $\mathbf X^{(i')} \in \mathcal{SU}(d_{i'})$, define
\begin{align}
	\mathbf H^{(i,i')}_{j_i,j_{i'}}[\mathbf X^{(i')}]
	:=
	D_{\mathbf U^{(i')}_{j_{i'}}}\mathbf G^{(i)}_{j_i}(z^\star)[\mathbf X^{(i')}].
\end{align}
To evaluate this derivative, we perturb the action $\mathbf U^{(i')}_{j_{i'}}$ along the tangent direction $\mathbf X^{(i')}$ according to
$
\mathbf U^{(i')}_{j_{i'}}(t)
=
\exp\!\bigl(t\mathbf X^{(i')}\bigr)\mathbf U^{(i')}_{j_{i'}}(0),
$
and differentiate with respect to $t$ at $t=0$.
\\
Using Eq.~(S35) and the mixed-strategy weights, we obtain the following expression at the fixed point $z^\star$ for $i\neq i'$:
\begin{align}
	\mathbf H^{(i,i')}_{j_i,j_{i'}}[\mathbf X^{(i')}]
	&=
	p^{(i)}_{j_i}p^{(i')}_{j_{i'}}
	\sum_{\mathbf j_{-(i,i')}}
	\left(\prod_{k\neq i,i'} p^{(k)}_{j_k}\right)
	\operatorname{Tr}_{-i}
	\left(
	\left[
	\mathbf R_i,
	\left[
	\mathbf X^{(i')}\otimes \mathbf I^{(-i')},
	\rho^\star_{j_i,j_{i'},\mathbf j_{-(i,i')}}
	\right]
	\right]
	\right).
\end{align}
Similarly, for $i=i'$, we have
\begin{align}
	\mathbf H^{(i,i)}_{j_i,j_i'}[\mathbf X^{(i)}]
	&=
	\delta_{j_i j_i'} p^{(i)}_{j_i}
	\sum_{\mathbf j_{-i}}
	\left(\prod_{k\neq i} p^{(k)}_{j_k}\right)
	\operatorname{Tr}_{-i}
	\left(
	\left[
	\mathbf R_i,
	\left[
	\mathbf X^{(i)}\otimes \mathbf I^{(-i)},
	\rho^\star_{j_i,\mathbf j_{-i}}
	\right]
	\right]
	\right).
\end{align}
To obtain a matrix representation of this linear operator, we choose orthonormal bases
\begin{align}
	\{\mathbf E^{(i)}_1,\dots,\mathbf E^{(i)}_{d_i^2-1}\}\subset \mathcal{SU}(d_i),
	\qquad \forall i,
\end{align}
with respect to the chosen inner product $\langle \mathbf A,\mathbf B\rangle$. Then the matrix entries of the operator $\mathbf H^{(i,i')}_{j_i,j_{i'}}$ are given by
\begin{align}
	\bigl[\mathbf H^{(i,i')}_{j_i,j_{i'}}\bigr]_{ab}
	=
	\left\langle
	\mathbf E_a^{(i)},
	\mathbf H^{(i,i')}_{j_i,j_{i'}}[\mathbf E_b^{(i')}]
	\right\rangle 
	=
	-\operatorname{Tr}\!\left(
	\mathbf E_a^{(i)} \mathbf H^{(i,i')}_{j_i,j_{i'}}[\mathbf E_b^{(i')}]
	\right),
\end{align}
where the last identity uses the trace inner product on $\mathcal{SU}(d_i)$. After stacking over all indices, one obtains the block $\mathbf H_{ii'}$.
\\
\subsubsection{Construction of $\mathbf B_{ii'}$}
We now compute the derivative of the action gradient with respect to the probability variable $p^{(i')}_{j_{i'}}$. Define
\begin{align}
	\mathbf B^{(i,i')}_{j_i,j_{i'}}
	:=
	D_{p^{(i')}_{j_{i'}}}\mathbf G^{(i)}_{j_i}(z^\star).
\end{align}
\\
For $i'\neq i$, differentiating with respect to $p^{(i')}_{j_{i'}}$ gives
\begin{align}
	\mathbf B^{(i,i')}_{j_i,j_{i'}}
	&=
	p^{(i)}_{j_i}
	\sum_{\mathbf j_{-(i,i')}}
	\left(\prod_{k\neq i,i'} p^{(k)}_{j_k}\right)
	\operatorname{Tr}_{-i}
	\left(
	\left[
	\mathbf R_i, 
	\rho^\star_{j_i,j_{i'},\mathbf j_{-(i,i')}}
	\right]
	\right).
\end{align}
\\
For $i'=i$, only the prefactor $p^{(i)}_{j_i}$ is differentiated, and therefore
\begin{align}
	\mathbf B^{(i,i)}_{j_i,j_i'}
	&=
	\delta_{j_i j_i'}
	\sum_{\mathbf j_{-i}}
	\left(\prod_{k\neq i} p^{(k)}_{j_k}\right)
	\operatorname{Tr}_{-i}
	\left(
	\left[
	\mathbf R_i, 
	\rho^\star_{j_i,\mathbf j_{-i}}
	\right]
	\right).
\end{align}
\\
The coordinates of $\mathbf B^{(i,i')}_{j_i,j_{i'}}$ with respect to the basis $\{\mathbf E_a^{(i)}\}$ are
\begin{align}
	\bigl[\mathbf B^{(i,i')}_{j_i,j_{i'}}\bigr]_a
	=
	\left\langle
	\mathbf E_a^{(i)},
	\mathbf B^{(i,i')}_{j_i,j_{i'}}
	\right\rangle 
	=
	-\operatorname{Tr}\!\left(
	\mathbf E_a^{(i)} 
	\mathbf B^{(i,i')}_{j_i,j_{i'}}
	\right).
\end{align}
Finally, after stacking over all indices, one obtains the block $\mathbf B_{ii'}$.
\\
\subsection*{Construction of $D\mathcal T_i^{p}(z^\star)$}
Since $\mathcal T_i^{p}$ modifies only the probability vector of player $i$, its differential is the identity on all blocks except the probability-variable block of player $i$.
Because the probability update is defined by the softmax function, its linearization at the fixed point $z^\star$ is
\begin{align}
	\delta \mathbf p^{(i)+}
	=
	\mathbf S_i
	\left(
	\sum_{i'=1}^N \mathbf C_{ii'} \delta z^{(i')}_{U}
	+
	\sum_{i'=1}^N \mathbf D_{ii'} \delta \mathbf p^{(i')}
	\right),
\end{align}
where
\begin{align}
	\mathbf S_i
	=
	\frac{1}{T}
	\Bigl(
	\operatorname{Diag}(\mathbf p^{(i)\star})
	-
	\mathbf p^{(i)\star}\mathbf p^{(i)\star T}
	\Bigr),
\end{align}
and
\begin{align}
	\mathbf C_{ii'}=D_{z^{(i')}_{U}}\boldsymbol \ell^{(i)}(z^\star),
	\qquad
	\mathbf D_{ii'}=D_{\mathbf p^{(i')}}\boldsymbol \ell^{(i)}(z^\star).
\end{align}
Therefore, $D\mathcal T_i^{p}(z^\star)$ is the identity on every block except the row corresponding to the probability variables of player $i$. On that row, it is given by
\begin{align}
	\bigl[
	\mathbf S_i\mathbf C_{i1}, 
	\mathbf S_i\mathbf D_{i1}, 
	\dots, 
	\mathbf S_i\mathbf C_{iN}, 
	\mathbf S_i\mathbf D_{iN}
	\bigr].
\end{align}
In the following, we derive the blocks $\mathbf C_{ii'}$ and $\mathbf D_{ii'}$ in detail.
\\
\subsubsection{Construction of $\mathbf C_{ii'}$}
The derivative of the per-action payoff $\ell^{(i)}_{j_i}$ used in the probability update with respect to the action variable $\mathbf U^{(i')}_{j_{i'}}$, evaluated at the fixed point $z^\star$, is a scalar-valued linear functional of the perturbation $\mathbf X^{(i')}$.
\\
For $i\neq i'$,
\begin{align}
	C^{(i,i')}_{j_i,j_{i'}}[\mathbf X^{(i')}]
	&=
	p^{(i')}_{j_{i'}}
	\sum_{\mathbf j_{-(i,i')}}
	\left(\prod_{k\neq i,i'} p^{(k)}_{j_k}\right)
	\operatorname{Tr}
	\!\left(
	\mathbf R_i
	\left[
	\mathbf X^{(i')} \otimes \mathbf I^{(-i')},
	\rho^\star_{j_i,j_{i'},\mathbf j_{-(i,i')}}
	\right]
	\right).
\end{align}
Similarly, for $i'=i$,
\begin{align}
	C^{(i,i)}_{j_i,j_i'}[\mathbf X^{(i)}]
	&=
	\delta_{j_i j_i'}
	\sum_{\mathbf j_{-i}}
	\left(\prod_{k\neq i} p^{(k)}_{j_k}\right)
	\operatorname{Tr}
	\!\left(
	\mathbf R_i
	\left[
	\mathbf X^{(i)} \otimes \mathbf I^{(-i)},
	\rho^\star_{j_i,\mathbf j_{-i}}
	\right]
	\right).
\end{align}
Since $C^{(i,i')}_{j_i,j_{i'}}$ is a scalar-valued linear functional of the perturbation, it is represented with respect to the orthonormal basis $\{\mathbf E^{(i')}_b\}$ by the row vector
\begin{align}
	\bigl[C^{(i,i')}_{j_i,j_{i'}}\bigr]_b
	=
	C^{(i,i')}_{j_i,j_{i'}}[\mathbf E^{(i')}_b].
\end{align}
Stacking these row vectors over all actions $j_i$ and $j_{i'}$ yields the block $\mathbf C_{ii'}$ at the fixed point.
\\
\subsubsection{Construction of $\mathbf D_{ii'}$}
The derivative of the per-action payoff $\ell^{(i)}_{j_i}$ used in the probability update with respect to the probability variable $p^{(i')}_{j_{i'}}$ is
\begin{align}
	D^{(i,i')}_{j_i,j_{i'}}
	:=
	D_{p^{(i')}_{j_{i'}}}\ell^{(i)}_{j_i}(z^\star).
\end{align}
For $i'\neq i$, we have
\begin{align}
	D^{(i,i')}_{j_i,j_{i'}}
	&=
	\sum_{\mathbf j_{-(i,i')}}
	\left(\prod_{k\neq i,i'} p^{(k)}_{j_k}\right)
	\operatorname{Tr}
	\!\left(
	\mathbf R_i \rho^\star_{j_i,j_{i'},\mathbf j_{-(i,i')}}
	\right),
\end{align}
which is a scalar.
Since $\ell^{(i)}_{j_i}$ does not depend on the mixed probabilities of player $i$, we have
$
	D^{(i,i)}_{j_i,j_i'}=0.
$
Stacking these scalars over the indices yields the block $\mathbf D_{ii'}$ at the fixed point.
\\
Note that when $m_i=1$, then the probability simplex of player $i$ is trivial, so the probability update $\mathcal T_i^p$ is absent and
$
D\mathcal T_i(z^\star)=D\mathcal T_i^U(z^\star).
$
In particular, if $m_i=1$ for all players, then the one-sweep differential is constructed only from the action-step differentials, and no probability-related blocks exist.

\section{Additional Information on Experimental Results}
\label{append:exp_info}

In our experiments, we consider an outcome set $\Omega = \{\omega_1, \omega_2, \ldots, \omega_m\}$, where $m$ represents the total number of  outcomes, which remains the same for all players. We use the projective probability operator $\mathbf{P}_{\omega_j} = \ket{\omega_j}\bra{\omega_j}$.
Each player $i$ has a payoff vector $\mathbf{r}_i = (R(\omega_1), R(\omega_2), \ldots, R(\omega_m))$, which assigns a real-valued payoff to every possible outcome. By combining these payoff vectors with the outcome set, we construct the payoff operator for each game, as \begin{align}
    \mathbf{R}_i= \sum_{\omega \in \Omega} R_i(\omega) \mathbf{P}_{\omega}.
\end{align} 

We use the initial density matrix $\rho_0 = \ket{\psi_0}\bra{\psi_0}$, where $\ket{\psi_0}$ is the initial state vector defined separately for each game and include entanglement parameters.
Detailed descriptions of each game used in the experimental results are provided below

\textbf{Game 1:}
In the two player quantum Prisoner’s Dilemma, we follow Ref.~\cite{eisert1999quantum} with $N = 2$, $d_i = 2$.
The initial states, payoff outcomes are
\begin{subequations}
\begin{align}
&\quad \ket{\psi_0} 
= \cos{(\gamma/2)}\ket{00}+i\sin{(\gamma/2)}\ket{11},
\\
(C,C):&\quad  \ket{\omega_1} 
= \cos{(\gamma/2)}\ket{00}+i\sin{(\gamma/2)}\ket{11},
\\
(C,D):&\quad  \ket{\omega_2} 
= \cos{(\gamma/2)}\ket{01}+i\sin{(\gamma/2)}\ket{10},\\
(D,C):&\quad  \ket{\omega_3} 
= i\sin{(\gamma/2)}\ket{01} + \cos{(\gamma/2)}\ket{10}, \\
(D,D):&\quad  \ket{\omega_4} 
= i\sin{(\gamma/2)}\ket{00} + \cos{(\gamma/2)}\ket{11},
\end{align}
\end{subequations}
The payoff vectors corresponding to players 1 and 2 are
\begin{align}
    \mathbf{r}_1 = (3,0,5,1), \quad
    \mathbf{r}_2 = (3,5,0,1).
\end{align}

\textbf{Game 2:}
In this two player quantum game we use
\begin{subequations}
\begin{align}
\ket{\psi_0} = \cos{(\gamma/2)} \ket{00} + i \sin{(\gamma/2)} \ket{12},\\
\ket{\omega_1} = \cos{(\gamma/2)} \ket{00} + i \sin{(\gamma/2)} \ket{12},\\
\ket{\omega_2} = \cos{(\gamma/2)} \ket{01} + i \sin{(\gamma/2)} \ket{11},\\
\ket{\omega_3} = \cos{(\gamma/2)} \ket{02} + i \sin{(\gamma/2)} \ket{10},\\
\ket{\omega_4} =   i \sin{(\gamma/2)} \ket{02} + \cos{(\gamma/2)} \ket{10},\\
\ket{\omega_5} =  i \sin{(\gamma/2)} \ket{01} + \cos{(\gamma/2)} \ket{11},\\
\ket{\omega_6} =  i \sin{(\gamma/2)} \ket{00} + \cos{(\gamma/2)} \ket{12}.
\end{align}
\end{subequations}
The corresponding payoff vectors for the two players are given by
\begin{align}
    \mathbf{r}_1 = (4, 5, 0.5, 1, 1.15, 1.25), \quad
    \mathbf{r}_2 = (4.25, 0.52, 5.2, 1.1, 1.55, 1.9).
\end{align}

\textbf{Game 3:}
\begin{subequations}
\begin{align}
\ket{\psi_0} &= \cos{(\gamma/2)} \ket{00} + i \sin{(\gamma/2)} \ket{22},\\
\ket{\omega_1} &= \cos{(\gamma/2)} \ket{00} + i \sin{(\gamma/2)} \ket{22},\\
\ket{\omega_2} &= \cos{(\gamma/2)} \ket{01} + i \sin{(\gamma/2)} \ket{21},\\
\ket{\omega_3} &= \cos{(\gamma/2)} \ket{02} + i \sin{(\gamma/2)} \ket{20},\\
\ket{\omega_4} &= \cos{(\gamma/2)} \ket{10} + i \sin{(\gamma/2)} \ket{12},\\
\ket{\omega_5} &= e^{i \gamma/2} \ket{11}, \\
\ket{\omega_6} &=  i \sin{(\gamma/2)} \ket{10}+\cos{(\gamma/2)} \ket{12},\\
\ket{\omega_7} &=i \sin{(\gamma/2)} \ket{02} +  \cos{(\gamma/2)} \ket{20},\\
\ket{\omega_8} &=  i \sin{(\gamma/2)}  \ket{01} +\cos{(\gamma/2)}\ket{21},\\
\ket{\omega_9} &=i \sin{(\gamma/2)}  \ket{00} +  \cos{(\gamma/2)}\ket{22},
\end{align}
\end{subequations}
and the corresponding payoff vectors for the two players are given by
\begin{align}
    \mathbf{r}_1 = (4, 5, 0.5, 1, 1.15, 1.25, 2, 11, 4), \quad
    \mathbf{r}_2 = (4.25, 11, 5.2, 6.1, 1.55, 1.9, 3, 2, 4).
\end{align}

\textbf{Game 4: Three-Player Prisoner’s Dilemma:}\\
In the three-player quantum Prisoner’s Dilemma game with $N = 3$ and $d_i = 2$ for all players, the initial states and payoff outcomes are as follows
\begin{subequations}
\begin{align}
\ket{\psi_0} &= \cos{(\gamma/2)} \ket{000} + i \sin{(\gamma/2)} \ket{111},\\
\ket{\omega_1} &= \cos{(\gamma/2)} \ket{000} + i \sin{(\gamma/2)} \ket{111},\\
\ket{\omega_2} &= \cos{(\gamma/2)} \ket{001} + i \sin{(\gamma/2)} \ket{110},\\
\ket{\omega_3} &= \cos{(\gamma/2)} \ket{010} + i \sin{(\gamma/2)} \ket{101},\\
\ket{\omega_4} &= \cos{(\gamma/2)} \ket{011} + i \sin{(\gamma/2)} \ket{100},\\
\ket{\omega_5} &=  i \sin{(\gamma/2)} \ket{011}+\cos{(\gamma/2)} \ket{100},\\
\ket{\omega_6} &=i \sin{(\gamma/2)} \ket{010} +  \cos{(\gamma/2)} \ket{101},\\
\ket{\omega_7} &=  i \sin{(\gamma/2)}  \ket{001} +\cos{(\gamma/2)}\ket{110},\\
\ket{\omega_8} &=i \sin{(\gamma/2)}  \ket{000} +  \cos{(\gamma/2)}\ket{111}.
\end{align}
\end{subequations}
The payoff vectors are
\begin{align}
\mathbf{r}_1 = (3,  2, 2, 0, 5, 4, 4, 1),\quad
\mathbf{r}_2 = (3,  2, 5, 4, 2, 0, 4, 1),\quad
\mathbf{r}_3 = (3,  5, 2, 4, 2, 4, 0, 1).
\end{align}
\end{document}